\documentclass[11pt]{article}

\usepackage[a4paper,margin=1in]{geometry}
\usepackage{amsmath,amssymb,amsthm,mathtools}
\usepackage{booktabs}
\usepackage{enumitem}
\usepackage{hyperref}

\hypersetup{
  colorlinks=true,
  linkcolor=blue,
  citecolor=blue,
  urlcolor=blue
}

\newtheorem{theorem}{Theorem}[section]
\newtheorem{lemma}[theorem]{Lemma}
\newtheorem{proposition}[theorem]{Proposition}
\newtheorem{corollary}[theorem]{Corollary}
\newtheorem{example}[theorem]{Example}
\theoremstyle{definition}
\newtheorem{definition}[theorem]{Definition}
\theoremstyle{remark}
\newtheorem{remark}[theorem]{Remark}

\makeatletter
\renewenvironment{proof}[1][\proofname]{%
  \par
  \pushQED{\qed}%
  \normalfont
  \topsep6\p@\@plus6\p@\relax
  \trivlist
  \item[\hskip\labelsep\bfseries #1\@addpunct{.}]\ignorespaces
}{%
  \popQED
  \endtrivlist
  \@endpefalse
}
\makeatother

\newcommand{\Z}{\mathbb Z}
\newcommand{\Q}{\mathbb Q}

\newcommand{\ii}{\mathrm i}
\newcommand{\Cay}{\operatorname{Cay}}

\title{Perfect State Transfer on Oriented Circulant Graphs: A Complete Classification}
\author{
	Xingkun Song$^{1,2}$
\thanks{Email: \href{mailto:xksong@126.com}{xksong@126.com}}\\
	{\small $^{1}$ School of Mathematics and Statistics, Qinghai Minzu University,}\\
	{\small Xining, Qinghai 810007, P.R. China}\\[1ex]
	{\small $^{2}$ Qinghai Institute of Applied Mathematics,}\\
	{\small Xining, Qinghai 810007, P.R. China}\\[1ex]
}
\date{}

\begin{document}

\maketitle

\begin{abstract}
The continuous-time quantum walk on an oriented circulant graph is determined
by the Fourier eigenvalues of its Hermitian adjacency matrix.  We classify
perfect state transfer (PST) between distinct vertices in every nonempty
oriented circulant graph.  We show that each such graph is described by an
odd primitive quadratic Dirichlet character of conductor $\Delta$, a set of
gcd-classes, and a choice between the two orientations of each selected
class.  For a graph of order $n$, we derive an explicit formula for every
Fourier eigenvalue without assuming that $n/\Delta$ is coprime to $\Delta$.
We prove that PST occurs only for $\Delta\in\{3,4,8\}$ and give necessary and
sufficient conditions on the connection set for each conductor.  Equivalently,
the square-free radicands of oriented circulant graphs with PST are exactly
$1$, $2$, and $3$.  More generally, when
$\lambda_j=\sqrt{D}\eta_j$ with $\eta_j\in\mathbb{Z}$, congruences satisfied
by the integers $\eta_j$ determine all PST pairs and times, the minimum
period, and the largest vertex sets supporting multiple state transfer (MST).
In this class, pretty good state transfer is equivalent to PST.  We also
determine the connected orders and enumerate the resulting graphs.
\end{abstract}



\section{Introduction}

Let $\Gamma=\Cay(\Z_n,C)$ be a circulant graph.  Its spectrum is the discrete
Fourier transform of the connection set $C$ \cite{Davis1979}; consequently,
its arithmetic is controlled by the cyclotomic field generated by the
relevant roots of unity.  Bridges and Mena, and later So, characterized
integral undirected circulant graphs as unions of gcd-classes
\cite{BridgesMena1979,So2006}.  Such graphs provide a well-studied setting
for periodic quantum walks and perfect state transfer (PST)
\cite{SaxenaSeveriniShparlinski2007,BasicPetkovicStevanovic2009}.

Quantum state transfer was introduced for spin networks by Bose
\cite{Bose2003}, and its graph-theoretic formulation is surveyed in
\cite{Godsil2012}.  For an oriented graph, the continuous-time walk is
defined from the Hermitian adjacency matrix \cite{GodsilLato2020}.  PST on
abelian Cayley graphs has been studied in several settings
\cite{TanFengCao2019}.  Kadyan and Bhattacharjya characterized integral mixed
Cayley graphs over abelian groups and, in the cyclic case, gave a
connection-set description and a Ramanujan-sum expression for the spectrum
\cite{KadyanBhattacharjya2021,KadyanBhattacharjya2023}.  Song then
characterized PST and multiple state transfer (MST) in integral oriented
circulant graphs \cite{Song2024IOC}.  The corresponding integral mixed
circulant graphs were treated in \cite{SongLin2024Mixed}; see also
\cite{SongLin2026Mixed}.

The residue classes modulo $4$ used in the integral theory are the two
classes determined by the odd quadratic character $\chi_4$.  More
generally, an odd primitive quadratic Dirichlet character $\chi$ of conductor
$\Delta$ has quadratic Gauss sum $\ii\sqrt{\Delta}$
\cite[Chapter~9]{IrelandRosen}.  Godsil and Zhang characterized the oriented
and signed Cayley graphs on finite abelian groups whose eigenvalues are
integral multiples of a square root
\cite[Theorem~8.1]{GodsilZhang2024}.  Their result determines the admissible
connection sets but does not give, in the notation needed here, the
eigenvalue attached to each Fourier character.  These individual eigenvalues
are needed because state transfer is governed by simultaneous phase
equations.

Godsil and Lato proved that periodicity in a connected oriented graph forces
a square-root condition on the eigenvalues in the support of a vertex
\cite[Theorem~6.1]{GodsilLato2020}.  The cyclic specialization of
\cite[Theorem~8.1]{GodsilZhang2024} is recorded in
Theorem~\ref{thm:cyclic-classification}, where it gives an explicit
description of the connection sets of oriented circulant graphs with
square-root spectra.

Chan, Pantangi, Razafimahatratra, and Sin recently studied PST and MST in
oriented normal Cayley graphs \cite{ChanPantangiRazafimahatratraSin2026}.
Their Corollary~2.3 restricts the cardinality of a PST equivalence class to
$2$, $3$, $4$, or $6$, while their Theorem~4.1 excludes cardinality $6$ for
oriented normal Cayley graphs on solvable groups.  Thus an oriented
circulant graph cannot support MST on more than four vertices.  The purpose
of the present paper is more specific: for cyclic groups we determine exactly
which connection sets admit transfer, all pairs of vertices joined by PST,
all PST times, and the resulting enumeration formulas.
Their character-theoretic condition for oriented normal Cayley graphs
\cite[Theorem~3.3]{ChanPantangiRazafimahatratraSin2026} provides a general
group-theoretic counterpart to the Fourier congruences used below.

We state the classification before developing its proof.  We use the notation
of Definition~\ref{def:main-graph} for the corresponding oriented circulant
graphs.  For each
$\Delta\in\{3,4,8\}$, let $\chi_\Delta$ denote the odd primitive quadratic
character of conductor $\Delta$.  Given a positive integer $m$ and a set
$X$ of divisors of $m$, put
\[
   p=
   \begin{cases}
      3,&\Delta=3,\\
      2,&\Delta\in\{4,8\},
   \end{cases}
   \qquad
   m=p^a u,\quad (p,u)=1,
\]
and, for $0\le r\le a$, put
\[
   X_r=X\cap\{p^{a-r}e:e\mid u\}.
\]

\begin{theorem}\label{thm:introduction-classification}
Let $\Gamma$ be a nonempty oriented circulant graph.  Then $\Gamma$ has PST
between distinct vertices if and only if there are
$\Delta\in\{3,4,8\}$, a positive integer $m$, a set $X$ of divisors of $m$,
and a map $\sigma:X\to\{\pm1\}$ such that
\[
   \Gamma=\mathcal I_{\Delta m,\chi_\Delta}(X,\sigma)
\]
and the corresponding row of the following table holds.  Here
$\mathcal T$ is the transfer subgroup: it consists of $0$ and all vertices
joined to $0$ by PST.
\begin{center}
\begin{tabular}{c@{\qquad}l@{\qquad}c}
\toprule
$\Delta$&condition&$\mathcal T$\\
\midrule
$3$&$X_0=\{m\}$&$m\Z_{3m}$\\
$4$&$a=0,\ X_0=\{m\}$&$\{0,2m\}$\\
$4$&$a\ge1,\ X_0=\{m\},\ X_1\ne\{m/2\}$&$\{0,2m\}$\\
$4$&$a\ge1,\ X_0=\{m\},\ X_1=\{m/2\}$&$m\Z_{4m}$\\
$8$&$X_0=\{m\}$&$\{0,4m\}$\\
\bottomrule
\end{tabular}
\end{center}
\end{theorem}

The proof of Theorem~\ref{thm:introduction-classification} is completed in
Section~\ref{sec:classification}.  Theorem~\ref{thm:all-orders-classification}
also gives the least transfer times in the cases listed above.

The first main step is the explicit spectral calculation.
Theorem~\ref{thm:general-spectrum} gives the eigenvalue associated with every Fourier
character of $\mathcal I_{n,\chi}(X,\sigma)$.  No coprimality assumption is
made between $n/\Delta$ and $\Delta$, so the formula includes the ramified
orders needed for connected examples.  It is obtained by separating a
complete character sum into a quadratic Gauss sum and a Ramanujan sum.

The second main step concerns state transfer for an arbitrary oriented
circulant graph with eigenvalues in $\sqrt{D}\Z$.
Theorem~\ref{thm:transfer-structure} determines all pairs of vertices joined
by PST, all PST times, the minimum period, and the largest MST sets from the integer
eigenvalue coefficients.  It also shows that pretty good state transfer
(PGST) is equivalent to PST under this spectral hypothesis.

Applying the explicit eigenvalue formula to this transfer theorem yields the
complete classification.  Theorem~\ref{thm:global-conductor} first restricts
PST between distinct vertices to
\[
   \Delta\in\{3,4,8\}.
\]
Theorem~\ref{thm:all-orders-classification} proves the conditions stated in
Theorem~\ref{thm:introduction-classification} at every order.  The possible
square-free radicands are $1$, $2$, and $3$.  The conductor-$4$ part recovers
the integral oriented classification in \cite{Song2024IOC}, while conductors
$3$ and $8$ give the nonintegral cases.  The classification also yields all
transfer times, the connected orders, and exact enumeration formulas.

Section~2 develops the quadratic-character construction and the explicit
spectral formula.  Section~3 proves the general theorem that determines PST
vertices and times from the spectrum.
Section~4 establishes the conductor restriction and the complete
classification.  Section~5 gives coprime and prime-power specializations,
examples, further consequences, and enumeration formulas.  Section~6
concludes the paper.

\section{Quadratic-character circulant graphs and their spectra}

All graphs in this paper are finite and simple.  An oriented circulant graph
is called connected when its underlying undirected graph is connected.  For
every positive integer
$r$, write $\zeta_r=\exp(2\pi\ii/r)$.  We write $\Z_n$ additively and identify
its characters with
\[
   \varphi_j(x)=\zeta_n^{jx},\qquad 0\le j\le n-1.
\]
For a positive integer $N$, write $\Z_N^\times$ for the unit group of
$\Z_N$.
If $C\subseteq \Z_n$ satisfies $C\cap(-C)=\varnothing$, the oriented
circulant graph $\Cay(\Z_n,C)$ has Hermitian adjacency matrix $A$ given by
\[
   A_{uv}=
   \begin{cases}
      \ii, & v-u\in C,\\
      -\ii, & u-v\in C,\\
      0, & \text{otherwise}.
   \end{cases}
\]
The Fourier vectors are eigenvectors of $A$, and the eigenvalue belonging to
$\varphi_j$ is
\begin{equation}\label{eq:hermitian-eigenvalue}
   \lambda_j
   =
   \ii\sum_{c\in C}\left(\zeta_n^{jc}-\zeta_n^{-jc}\right).
\end{equation}

\subsection{Odd quadratic characters}

Let $\chi$ be a primitive quadratic Dirichlet character of conductor
$\Delta$.  The standard convention is that $\chi(a)=0$ when
$(a,\Delta)>1$.  Since $\chi$ is quadratic, its values on units are
$\pm1$.  The character is called odd if $\chi(-1)=-1$.  If $N$ is a
multiple of $\Delta$ and $a\in\Z_N$, then $\chi(a)$ means that $\chi$ is
evaluated at the image of $a$ in $\Z_\Delta$.

\begin{proposition}\label{prop:parity}
Let $N$ be a multiple of $\Delta$, and define
\[
   C_\chi^+(N)=\{a\in \Z_N^\times:\chi(a)=1\},\qquad
   C_\chi^-(N)=\{a\in \Z_N^\times:\chi(a)=-1\}.
\]
If $\chi$ is odd, then
\[
   -C_\chi^+(N)=C_\chi^-(N).
\]
In particular, $C_\chi^+(N)\cap (-C_\chi^+(N))=\varnothing$, so
$C_\chi^+(N)$ is an oriented connection set.
\end{proposition}

\begin{proof}
Let $a\in C_\chi^+(N)$.  Since $\chi$ is odd,
\[
   \chi(-a)=\chi(-1)\chi(a)=-1.
\]
Thus $-a\in C_\chi^-(N)$, proving one inclusion.  Interchanging the two
character classes gives the reverse inclusion.
\end{proof}

We shall use the standard correspondence between primitive quadratic
characters and fundamental discriminants, together with the quadratic
Gauss-sum evaluation; see \cite[Chapter~9]{IrelandRosen}.

\begin{proposition}\label{prop:delta}
There exists an odd primitive quadratic character of conductor $\Delta$ if
and only if $-\Delta$ is a fundamental discriminant.  Equivalently,
\[
   \Delta\equiv 3 \pmod 4 \quad\text{and $\Delta$ is square-free,}
\]
or
\[
   \Delta=4q,\qquad q\ \text{square-free and}\quad q\equiv1,2\pmod4.
\]
For such a character,
\[
   \tau(\chi):=\sum_{a=1}^{\Delta}\chi(a)\zeta_\Delta^a
   =\ii\sqrt{\Delta},
\]
after the standard choice of the positive square root.
\end{proposition}

\begin{proof}
Primitive quadratic characters correspond to fundamental discriminants, and
the character is odd precisely when the associated discriminant is negative.
The two alternatives in the statement are the usual classification of negative
fundamental discriminants.  The formula for $\tau(\chi)$ is the classical
quadratic Gauss-sum evaluation; see, for example,
\cite[Chapter~9]{IrelandRosen}.
\end{proof}

For a positive square-free integer $\delta$, set
\[
   \Delta(\delta)=
   \begin{cases}
      \delta, & \delta\equiv3\pmod4,\\
      4\delta, & \delta\equiv1,2\pmod4.
   \end{cases}
\]
Then
\[
   \sqrt{\Delta(\delta)}
   =
   \begin{cases}
      \sqrt{\delta}, & \delta\equiv3\pmod4,\\
      2\sqrt{\delta}, & \delta\equiv1,2\pmod4.
   \end{cases}
\]
In particular,
\[
   \Delta(1)=4,\qquad \Delta(2)=8,\qquad \Delta(3)=3.
\]

\begin{remark}\label{rem:conductor-radicand}
In this paper $\Delta$ denotes the conductor of the primitive quadratic
character, equivalently the absolute value of the associated negative
fundamental discriminant.  The square-free radicand is the square-free part
of $\Delta$.  In particular, radicand $2$ corresponds to the discriminant
$-8$ and conductor $8$, not conductor $2$.  Godsil and Zhang
\cite{GodsilZhang2024} use a real skew adjacency matrix, whereas the
Hermitian adjacency matrix used here is obtained by multiplication by
$\pm\ii$, according to the arc convention.  Their eigenvalues in
$\sqrt{-\delta}\Z$ therefore correspond, up to an overall sign, to Hermitian
eigenvalues in $\sqrt\delta\Z$.
\end{remark}

\subsection{Construction from gcd-classes}

For the remainder of the paper, unless stated otherwise, fix an odd primitive
quadratic character $\chi$ of conductor $\Delta$.  Let $n$ be a multiple of
$\Delta$.  The admissible gcd-classes are the classes in which the character
can be evaluated after division by the gcd.  Accordingly, for a divisor $d$ of
$n$ such that $\Delta\mid n/d$, define
\[
   \Theta_{n,\chi}^{\varepsilon}(d)
   =
   \{\,dk\in \Z_n:\ k\in \Z_{n/d}^{\times},\ \chi(k)=\varepsilon\,\},
   \qquad \varepsilon\in\{\pm1\}.
\]
The set $\Theta_{n,\chi}^{\varepsilon}(d)$ lies in the gcd-class
\[
   \Theta_n(d)=\{x\in \Z_n:\gcd(x,n)=d\}.
\]
By Proposition~\ref{prop:parity},
\[
   -\Theta_{n,\chi}^{\varepsilon}(d)=\Theta_{n,\chi}^{-\varepsilon}(d).
\]

\begin{definition}\label{def:main-graph}
Let $X$ be a set of divisors $d$ of $n$ satisfying $\Delta\mid n/d$, and let
$\sigma:X\to\{\pm1\}$.  Define
\[
   \mathcal I_{n,\chi}(X,\sigma)
\]
to be the oriented circulant graph on $\Z_n$ with connection set
\[
   C(X,\sigma)
   =
   \bigcup_{d\in X}\Theta_{n,\chi}^{\sigma(d)}(d).
\]
\end{definition}

Let $\Gamma=\mathcal I_{n,\chi}(X,\sigma)$.  The sets in the displayed union
are disjoint because they belong to distinct gcd-classes.  Thus $d$ specifies
the gcd-class and $\sigma(d)$ selects one of its two character classes.  By
Proposition~\ref{prop:parity}, each selected class is disjoint from its
negative; hence $\Gamma$ is an oriented circulant graph.

\begin{remark}\label{rem:sign-reversal}
Let $\Gamma=\mathcal I_{n,\chi}(X,\sigma)$ be an oriented circulant graph,
and let
$u\in\Z_n^\times$.  Multiplication by $u$ is an isomorphism
\[
   \mathcal I_{n,\chi}(X,\sigma)
   \cong
   \mathcal I_{n,\chi}(X,\chi(u)\sigma),
\]
where $(\chi(u)\sigma)(d)=\chi(u)\sigma(d)$.  In particular,
$\mathcal I_{n,\chi}(X,\sigma)$ and
$\mathcal I_{n,\chi}(X,-\sigma)$ are isomorphic.  Write $A(X,\sigma)$ and
$U_{X,\sigma}(t)=\exp(\ii tA(X,\sigma))$ for the Hermitian adjacency matrix
and transition matrix of $\mathcal I_{n,\chi}(X,\sigma)$, respectively.
Then
\[
   A(X,-\sigma)=-A(X,\sigma),
\]
and hence
\[
   U_{X,-\sigma}(t)=U_{X,\sigma}(-t).
\]
Consequently, $\mathcal I_{n,\chi}(X,\sigma)$ and
$\mathcal I_{n,\chi}(X,-\sigma)$ have the same pairs of vertices joined by
PST, the same periods, and the same maximum cardinality of a set supporting
MST.
Indeed, multiplication by $u$ preserves every gcd-class and sends
\[
   \Theta_{n,\chi}^{\sigma(d)}(d)
   \quad\text{onto}\quad
   \Theta_{n,\chi}^{\chi(u)\sigma(d)}(d).
\]
This proves the isomorphism.  Taking $u=-1$ and using $\chi(-1)=-1$ gives
the assertion for $-\sigma$.  The connection set of
$\mathcal I_{n,\chi}(X,-\sigma)$ is the negative of the connection set of
$\mathcal I_{n,\chi}(X,\sigma)$, so its Hermitian adjacency matrix is
$-A(X,\sigma)$.  The identity for the transition matrices follows.  If
$U_{X,\sigma}(t)=P_s$, then
$U_{X,-\sigma}(t)=P_{-s}$.  Let $k$ be the order of $s$ in $\Z_n$.
Since $U_{X,\sigma}(kt)=P_s^k=I$, we also have
\[
   U_{X,-\sigma}((k-1)t)
   =U_{X,\sigma}(-(k-1)t)
   =U_{X,\sigma}(t)
   =P_s.
\]
Thus every translation attained by one graph is attained by the other;
the least times for the two opposite translations may be interchanged.
The assertions about PST pairs, periods, and MST now follow.
\end{remark}

\begin{proposition}\label{prop:connectedness}
Let $\Gamma=\mathcal I_{n,\chi}(X,\sigma)$ be a nonempty oriented circulant
graph, and put
\[
   \kappa=\gcd\{d:d\in X\},\qquad
   n'=\frac n\kappa,\qquad
   X'=\left\{\frac d\kappa:d\in X\right\}.
\]
Define $\sigma'(d/\kappa)=\sigma(d)$.  Then $\Gamma$ has exactly $\kappa$
connected components, each isomorphic to the connected oriented circulant
graph
\[
   \Gamma'=\mathcal I_{n',\chi}(X',\sigma').
\]
Consequently, $\Gamma$ is connected if and only if $\kappa=1$, and the
characteristic polynomial of $\Gamma$ is the $\kappa$-th power of that of
$\Gamma'$.
\end{proposition}

\begin{proof}
We first compare the gcd-classes defining $\Gamma$ and $\Gamma'$.  For
$d=\kappa d'$, we have
\[
   \frac{n'}{d'}=\frac nd,
\]
so the divisors in $X'$ satisfy the condition in Definition~\ref{def:main-graph}.  Moreover,
\[
   \Theta_{n,\chi}^{\sigma(d)}(d)
   =\kappa\Theta_{n',\chi}^{\sigma'(d')}(d'),
\]
where multiplication by $\kappa$ identifies $\Z_{n'}$ with the subgroup
$\kappa\Z_n$ of $\Z_n$.  Thus
\[
   C(X,\sigma)=\kappa C(X',\sigma').
\]
It follows that every edge of $\Gamma$ joins two vertices in the same residue
class modulo $\kappa$.  To identify the subgraphs on these classes, fix
$c\in\{0,\ldots,\kappa-1\}$.  The map
\[
   c+\kappa y\longmapsto y
\]
identifies the subgraph on $c+\kappa\Z_n$ with $\Gamma'$.  Since
every element of $\Theta_{n',\chi}^{\sigma'(d')}(d')$ has gcd $d'$ with
$n'$, each selected gcd-class generates the additive subgroup $d'\Z_{n'}$.
Since
$\gcd\{d':d'\in X'\}=1$, the full connection set of $\Gamma'$ generates
$\Z_{n'}$, so $\Gamma'$ is connected.  These $\kappa$ cosets are therefore
the connected components of $\Gamma$.  Finally, the Hermitian adjacency
matrix of $\Gamma$ is permutation similar to the direct sum of $\kappa$
copies of the Hermitian adjacency matrix of $\Gamma'$, which gives the
assertion about the characteristic polynomial.
\end{proof}

\subsection{Relation with integral oriented circulant graphs}

The conductor-$4$ case connects the present notation with the integral
oriented circulant graphs.  Let $\Delta=4$.  For $r\in\{1,3\}$ and every divisor $d$ of $n$ satisfying
$4\mid n/d$, define
\[
   \Theta_n^r(d)
   =
   \{\,dk\in\Z_n:\ k\in\Z_{n/d}^{\times},\ k\equiv r\pmod4\,\}.
\]
The character $\chi_4$ separates the residue classes $1$ and $3$ modulo $4$,
so
\[
   \Theta_{n,\chi_4}^{+1}(d)=\Theta_n^1(d),\qquad
   \Theta_{n,\chi_4}^{-1}(d)=\Theta_n^3(d).
\]
Hence the connection set in Definition~\ref{def:main-graph} becomes
\[
   \bigcup_{d\in X}
   \begin{cases}
      \Theta_n^1(d), & \sigma(d)=1,\\
      \Theta_n^3(d), & \sigma(d)=-1.
   \end{cases}
\]
Thus the graphs $\mathcal I_{n,\chi_4}(X,\sigma)$ are precisely the integral
oriented circulant graphs studied in
\cite{KadyanBhattacharjya2021,KadyanBhattacharjya2023,Song2024IOC}.  The
cyclic connection-set description and Ramanujan-sum formula are given
explicitly in \cite{KadyanBhattacharjya2023}.  PST and MST in this integral
oriented family were characterized in \cite{Song2024IOC}, and the integral
mixed circulant case was characterized in \cite{SongLin2024Mixed}.  We retain
the conductor-$4$ specialization of the spectral formula below because it is
equivalent to, but arranged differently from, the piecewise $2$-adic formula
in \cite{Song2024IOC}.

\subsection{Gauss sums, Ramanujan sums, and the spectrum}

We next evaluate the Fourier contribution of a single selected character
class.  For $r\ge1$, let
\[
   c_r(j)=\sum_{a\in \Z_r^\times}\zeta_r^{ja}
\]
be the Ramanujan sum.  We use its standard properties; see, for example,
\cite{Apostol1976}.  Let $\varphi$ denote Euler's totient function and let
$\mu$ denote the M\"obius function.  If $g=(r,j)$,
then
\begin{equation}\label{eq:ramanujan-divisor}
   c_r(j)=\sum_{e\mid g}e\mu(r/e)
   =\mu(r/g)\frac{\varphi(r)}{\varphi(r/g)},
\end{equation}
and $c_1(j)=1$.  The next definition records the contribution from either of
the two classes determined by $\chi$.

\begin{lemma}\label{lem:ramanujan-independent}
Let $\ell$ be a prime with $\ell\nmid m$.  The functions
$j\mapsto c_q(j)$ on $\Z_m$, indexed by the divisors $q\mid m$, are linearly
independent over the field $\mathbb F_\ell$ with $\ell$ elements.
\end{lemma}

\begin{proof}
Let
\[
   F(j)=\sum_{q\mid m}a_q c_q(j)
\]
vanish in $\mathbb F_\ell$ for every $j$.  To express the coefficients in a
form suitable for M\"obius inversion, the divisor formula for Ramanujan sums
gives
\[
   F(j)=\sum_{e\mid(j,m)}eB_e,
   \qquad
   B_e=\sum_{\substack{q:\ e\mid q\mid m}}a_q\mu(q/e).
\]
Now let $(j,m)$ range over the divisors of $m$.  M\"obius inversion gives
$eB_e=0$ for every $e\mid m$.  Since $\ell\nmid m$, each such $e$ is
invertible in $\mathbb F_\ell$, and hence $B_e=0$.  Applying M\"obius inversion
once more to the definition of $B_e$ yields $a_q=0$ for every $q\mid m$.
\end{proof}

\begin{lemma}\label{lem:unit-reduction}
If $N\mid n$, then the natural reduction map
\[
   \Z_n^\times\longrightarrow\Z_N^\times
\]
is surjective.
\end{lemma}

\begin{proof}
Write
\[
   n=\prod_p p^{a_p},
   \qquad
   N=\prod_p p^{b_p},
   \qquad 0\le b_p\le a_p.
\]
Let $x$ be a unit modulo $N$.  For each prime $p$ with $b_p>0$, choose a
unit modulo $p^{a_p}$ that reduces to $x$ modulo $p^{b_p}$; for
$b_p=0$, choose the residue $1$ modulo $p^{a_p}$.  The Chinese remainder
theorem combines these choices into a unit modulo $n$ that reduces to $x$
modulo $N$.
\end{proof}

\begin{definition}
Let $N$ be a multiple of $\Delta$.  For $\varepsilon\in\{\pm1\}$ and
$j\in\Z$, define
\[
   s_{N,\chi}^{\varepsilon}(j)
   =
   \ii\sum_{\substack{a\in\Z_N^\times\\ \chi(a)=\varepsilon}}
      \left(\zeta_N^{ja}-\zeta_N^{-ja}\right).
\]
\end{definition}

\begin{lemma}\label{lem:half-full}
If $\chi$ is odd, then
\[
   s_{N,\chi}^{\varepsilon}(j)
   =
   \ii\,\varepsilon\sum_{a\in\Z_N^\times}\chi(a)\zeta_N^{ja}.
\]
\end{lemma}

\begin{proof}
Put $C_\varepsilon=\{a\in\Z_N^\times:\chi(a)=\varepsilon\}$.  The oddness
of $\chi$ gives $-C_\varepsilon=C_{-\varepsilon}$.  Consequently,
\[
\begin{aligned}
\sum_{a\in\Z_N^\times}\chi(a)\zeta_N^{ja}
  &=\varepsilon\sum_{a\in C_\varepsilon}\zeta_N^{ja}
    -\varepsilon\sum_{a\in C_\varepsilon}\zeta_N^{-ja}  \\
  &=\varepsilon\sum_{a\in C_\varepsilon}\left(\zeta_N^{ja}-\zeta_N^{-ja}\right).
\end{aligned}
\]
The right-hand side is $\varepsilon^{-1}s_{N,\chi}^{\varepsilon}(j)/\ii$.
Since $\varepsilon^2=1$, multiplication by $\ii\varepsilon$ gives the
claimed identity.
\end{proof}

For $q\ge1$, let $v_p(q)$ denote the exponent of the prime $p$ in $q$, and put
\[
   Q_\Delta(q)=\prod_{p\mid\Delta}p^{v_p(q)},
   \qquad
   M_\Delta(q)=\frac{q}{Q_\Delta(q)}.
\]
Thus $Q_\Delta(q)$ is the part of $q$ supported on the primes of $\Delta$,
whereas $(M_\Delta(q),\Delta)=1$.  This decomposition separates the ramified
and coprime parts of the character sum.

\begin{theorem}\label{thm:general-factorization}
Let $N=\Delta q$, with no coprimality assumption on $\Delta$ and $q$.  Then,
for every $j\in\Z$ and $\varepsilon\in\{\pm1\}$,
\[
   s_{N,\chi}^{\varepsilon}(j)
   =
   \begin{cases}
   -\varepsilon Q\chi(M)\chi(j/Q)c_M(j)\sqrt{\Delta},
      & Q\mid j,\\[4pt]
   0, & Q\nmid j,
   \end{cases}
\]
where $Q=Q_\Delta(q)$ and $M=M_\Delta(q)$.
\end{theorem}

\begin{proof}
Write $q=QM$.  By construction $(\Delta Q,M)=1$, so the Chinese remainder
theorem applies.  Choose integers $M_1$ and $D_1$ such that
\[
   MM_1\equiv1\pmod{\Delta Q},
   \qquad
   \Delta QD_1\equiv1\pmod M.
\]
Every $a\in\Z_N^\times$ has a unique representation
\[
   a\equiv uMM_1+v\Delta QD_1\pmod N,
   \qquad
   u\in\Z_{\Delta Q}^{\times},\quad v\in\Z_M^\times.
\]
Since $\chi(a)=\chi(u)$, the complete character sum in
Lemma~\ref{lem:half-full} factors as
\[
\begin{aligned}
   \sum_{a\in\Z_N^\times}\chi(a)\zeta_N^{ja}
   &=
   \left(
      \sum_{u\in\Z_{\Delta Q}^\times}
         \chi(u)\zeta_{\Delta Q}^{jM_1u}
   \right)
   \left(
      \sum_{v\in\Z_M^\times}
         \zeta_M^{jD_1v}
   \right)\\
   &=
   \left(
      \sum_{u\in\Z_{\Delta Q}^\times}
         \chi(u)\zeta_{\Delta Q}^{jM_1u}
   \right)c_M(j).
\end{aligned}
\]
The last equality follows because $D_1$ is a unit modulo $M$ and therefore
does not change the Ramanujan sum.  Every prime divisor of $Q$ divides
$\Delta$.  Consequently, the units modulo
$\Delta Q$ are precisely the residues $u=x+\Delta t$, where
$x\in\Z_\Delta^\times$ and $0\le t<Q$.  Therefore
\[
\sum_{u\in\Z_{\Delta Q}^\times}
   \chi(u)\zeta_{\Delta Q}^{jM_1u}
=
\sum_{x\in\Z_\Delta^\times}
   \chi(x)\zeta_{\Delta Q}^{jM_1x}
\sum_{t=0}^{Q-1}\zeta_Q^{jM_1t}.
\]
The second factor is a complete geometric sum, and $M_1$ is a unit modulo
$Q$.

\smallskip
\noindent\textbf{Case 1.} $Q\nmid j$.
The geometric sum vanishes, so the complete character sum and
$s_{N,\chi}^{\varepsilon}(j)$ are both zero.

\smallskip
\noindent\textbf{Case 2.} $Q\mid j$.
Write $j=Qh$.  The preceding display then becomes
\[
   Q\sum_{x\in\Z_\Delta^\times}
      \chi(x)\zeta_\Delta^{hM_1x}
   =Q\chi(h)\chi(M_1)\tau(\chi).
\]
Here we used the standard identity
\[
   \sum_{x\in\Z_\Delta^\times}\chi(x)\zeta_\Delta^{vx}
   =\chi(v)\tau(\chi)
\]
for a primitive quadratic character.  The identity also holds when
$(v,\Delta)>1$, in which case both sides are zero.  The congruence
$M M_1\equiv1\pmod\Delta$ and the quadratic nature of $\chi$ give
$\chi(M_1)=\chi(M)$.  Substituting this evaluation into the CRT
factorization, and then applying Lemma~\ref{lem:half-full} together with
$\tau(\chi)=\ii\sqrt{\Delta}$, gives the stated formula.
\end{proof}

\begin{example}\label{ex:basic-spectra}
Take $N=\Delta$ and $\varepsilon=1$ in
Theorem~\ref{thm:general-factorization}.  Then
\[
   s_{\Delta,\chi}^{+1}(j)=-\chi(j)\sqrt{\Delta}.
\]
For $\Delta=3$, $4$, and $8$, respectively, this gives
\[
\begin{aligned}
\Delta=3:\quad&
 C_\chi^+(3)=\{1\},&
(\lambda_0,\lambda_1,\lambda_2)&=(0,-\sqrt3,\sqrt3),\\
\Delta=4:\quad&
 C_\chi^+(4)=\{1\},&
(\lambda_0,\ldots,\lambda_3)&=(0,-2,0,2),\\
\Delta=8:\quad&
 C_\chi^+(8)=\{1,3\},&
(\lambda_0,\ldots,\lambda_7)&=(0,-\sqrt8,0,-\sqrt8,0,\sqrt8,0,\sqrt8).
\end{aligned}
\]
These direct Fourier evaluations also verify the sign in
Theorem~\ref{thm:general-factorization} for the three conductors that occur in the
state-transfer classification.
\end{example}

\begin{theorem}\label{thm:general-spectrum}
Let $\Gamma=\mathcal I_{n,\chi}(X,\sigma)$ be an oriented circulant graph.
For $d\in X$, put
\[
   q_d=\frac{n}{d\Delta},\qquad
   Q_d=Q_\Delta(q_d),\qquad M_d=M_\Delta(q_d).
\]
Then the Hermitian eigenvalues of $\Gamma$ are
\begin{equation}\label{eq:general-spectrum}
   \lambda_j
   =
   -\left(
      \sum_{\substack{d\in X\\Q_d\mid j}}
         \sigma(d)Q_d\chi(M_d)\chi(j/Q_d)c_{M_d}(j)
   \right)
      \sqrt{\Delta},
   \qquad 0\le j<n.
\end{equation}
In particular, every eigenvalue belongs to $\sqrt{\Delta}\Z$.
\end{theorem}

\begin{proof}
Fix $d\in X$.  The contribution of the gcd-class
$\Theta_{n,\chi}^{\sigma(d)}(d)$ is $s_{n/d,\chi}^{\sigma(d)}(j)$.  Since
$n/d=\Delta q_d$, Theorem~\ref{thm:general-factorization} evaluates this
contribution as the summand in \eqref{eq:general-spectrum}; it is zero when
$Q_d\nmid j$.  Summing the contributions of the disjoint gcd-classes indexed
by $d\in X$ gives \eqref{eq:general-spectrum}.
\end{proof}

Together, \eqref{eq:general-spectrum} and \eqref{eq:ramanujan-divisor}
express each eigenvalue through gcds, divisibility conditions, and values of
$\chi$.  The formula also gives the following spectral symmetry.

\begin{corollary}\label{cor:spectral-symmetry}
Under the hypotheses of Theorem~\ref{thm:general-spectrum},
\[
   \lambda_0=0,
   \qquad
   \lambda_{n-j}=-\lambda_j
   \quad (1\le j<n).
\]
Consequently, the multiplicities of $\theta$ and $-\theta$ are equal.  If
$n$ is even, then $\lambda_{n/2}=0$ as well.
\end{corollary}

\begin{proof}
In \eqref{eq:general-spectrum}, the divisibility condition $Q_d\mid j$ and
the Ramanujan sum are unchanged when $j$ is replaced by $-j$, whereas
$\chi(-j/Q_d)=-\chi(j/Q_d)$.  The assertion at $j=0$ follows because
$\chi(0)=0$.  When $n$ is even, the indices $n/2$ and $-n/2$ agree modulo
$n$, so the same symmetry gives $\lambda_{n/2}=0$.
\end{proof}

Let $\delta$ be the square-free part of $\Delta$.  Then $\Delta=\delta$ when
$\delta\equiv3\pmod4$, whereas $\Delta=4\delta$ when
$\delta\equiv1,2\pmod4$.  Theorem~8.1 of Godsil and Zhang
\cite{GodsilZhang2024} describes the admissible connection sets as unions of
orbits of an index-two subgroup of the unit group.  The following theorem
specializes that result to cyclic groups under the Hermitian adjacency
convention used here.

\begin{theorem}\label{thm:cyclic-classification}
Let $\chi$ be an odd primitive quadratic character of conductor $\Delta$,
and let $\Gamma=\Cay(\Z_n,C)$ be a nonempty oriented circulant graph.  All
Hermitian eigenvalues of $\Gamma$ belong to $\sqrt{\delta}\Z$ if and only if
$\Delta\mid n$ and
\[
   C=C(X,\sigma)
\]
for some nonempty set $X$ of divisors $d$ satisfying $\Delta\mid n/d$ and a
map $\sigma:X\to\{\pm1\}$.  In this case, the eigenvalues of $\Gamma$ belong
to $\sqrt{\Delta}\Z$.  If $\Delta=4\delta$, their coefficients relative to
$\sqrt\delta$ are even.
\end{theorem}

\begin{proof}
\noindent\textbf{Necessity.}
Suppose that all Hermitian eigenvalues of $\Gamma$ belong to
$\sqrt{\delta}\Z$.  Let $A$ be the Hermitian adjacency matrix of $\Gamma$ and put
$S=-\ii A$.  Since $C$ is nonempty, $S$ is a nonzero real skew-symmetric
circulant matrix.  It therefore has a nonzero eigenvalue of the form
$k\sqrt{-\delta}$ for some $k\in\Z$.  The eigenvalues of a circulant matrix
with integer entries and order $n$ belong to $\Q(\zeta_n)$, and hence
\[
   \Q(\sqrt{-\delta})\subseteq\Q(\zeta_n).
\]
The conductor of the quadratic field on the left is $\Delta$.  The
conductor-divisibility theorem for subfields of cyclotomic fields therefore
gives $\Delta\mid n$; see, for
example, \cite[Chapter~9]{IrelandRosen}.

We next describe the permitted gcd-classes.  Let
\[
   H=\{u\in\Z_n^\times:\chi(u)=1\}.
\]
Godsil and Zhang use the negative square-free radicand $-\delta$ for an
oriented graph.  After multiplication of their real skew adjacency matrix by
$\ii$, their characterization \cite[Theorem~8.1]{GodsilZhang2024} says that
the required connection sets are exactly the unions of $H$-orbits.  Consider
the gcd-class indexed by $d$ and put
$N=n/d$.  Let
\[
    \pi_N:\Z_n^\times\longrightarrow\Z_N^\times
\]
be the reduction map.  By Lemma~\ref{lem:unit-reduction}, this map is onto.
If $\Delta\mid N$, the character $\chi$ is defined modulo $N$ and
\[
   \pi_N(H)=\ker\!\left(\chi|_{\Z_N^\times}\right).
\]
Hence the gcd-class splits into the two orbits
$\Theta_{n,\chi}^{+1}(d)$ and $\Theta_{n,\chi}^{-1}(d)$.

Suppose now that $\Delta\nmid N$.  If $\pi_N(H)$ were proper, then it would
have index two in $\Z_N^\times$.  To see this, note that $\pi_N$ is onto and
$H$ has index two in $\Z_n^\times$.  Put
$K=\pi_N^{-1}(\pi_N(H))$.  Then
\[
   H\subseteq K,
\]
so $[\Z_n^\times:K]$ divides $[\Z_n^\times:H]=2$.  Since
$\pi_N(H)$ is proper and $\pi_N$ is onto, $K$ is proper.  Hence
$[\Z_n^\times:K]=2$, and
\[
   [\Z_n^\times:K]
   =
   [\Z_N^\times:\pi_N(H)]
   =2
   =
   [\Z_n^\times:H].
\]
It follows that $K=H$, and in particular $\ker\pi_N\subseteq H$.  Therefore
\[
   \psi_N(\pi_N(u))=\chi(u)\qquad (u\in\Z_n^\times)
\]
defines a quadratic character $\psi_N$ modulo $N$, with
$\chi=\psi_N\circ\pi_N$.  Thus the primitive character $\chi$ would be
induced by a character whose modulus divides $N$.  Its conductor $\Delta$
would then divide $N$, a contradiction.  Hence
$\pi_N(H)=\Z_N^\times$.  The gcd-class is a single $H$-orbit and is closed
under negation, so it cannot occur in an oriented connection set.  Thus
every allowable $C$ has the form
$C(X,\sigma)$.

\smallskip
\noindent\textbf{Sufficiency.}
Suppose that $\Delta\mid n$ and
$C=C(X,\sigma)$.  The same characterization of the $H$-orbits gives the
required square-root spectrum.  Finally, Theorem~\ref{thm:general-spectrum}
shows that the eigenvalues belong to $\sqrt{\Delta}\Z$ and gives the asserted
parity of the coefficients when $\Delta=4\delta$.
\end{proof}

\begin{remark}
The nonempty hypothesis in Theorem~\ref{thm:cyclic-classification} excludes
$\Cay(\Z_n,\varnothing)$.  This empty oriented circulant graph has zero
Hermitian spectrum, which lies in $\sqrt{\delta}\Z$ for every $\delta$, but it
does not admit PST between distinct vertices.
\end{remark}

\begin{corollary}\label{cor:coprime-spectrum}
Let $\Gamma=\mathcal I_{n,\chi}(X,\sigma)$ be an oriented circulant graph,
where $n=\Delta m$, $(\Delta,m)=1$, and $X$ is a set of divisors of $m$.
Then the Hermitian eigenvalues of $\Gamma$ have the form
\[
   \lambda_j=\eta_j\sqrt{\Delta},\qquad 0\le j\le n-1,
\]
where
\begin{equation}\label{eq:eta}
   \eta_j
   =
   -\chi(j)\sum_{d\in X}\sigma(d)\chi(m/d)c_{m/d}(j)
   \in\Z.
\end{equation}
\end{corollary}

\begin{proof}
For every $d\in X$, the integer $q_d=m/d$ is coprime to $\Delta$.  Hence
$Q_d=1$ and $M_d=m/d$ in Theorem~\ref{thm:general-spectrum}.  Substitution
in \eqref{eq:general-spectrum} gives \eqref{eq:eta}.
\end{proof}

\begin{corollary}\label{cor:primitive-spectrum}
Let $\chi$ be an odd primitive quadratic character of conductor $\Delta$, and
let
\[
   \Gamma_\chi=\Cay(\Z_\Delta,C_\chi^+(\Delta)).
\]
Then $\Gamma_\chi$ is a connected oriented circulant graph.  Its Hermitian eigenvalues are
\[
   \lambda_j=-\chi(j)\sqrt{\Delta}\qquad (0\le j<\Delta).
\]
Its characteristic polynomial is
\[
   x^{\Delta-\varphi(\Delta)}
   (x^2-\Delta)^{\varphi(\Delta)/2}.
\]
In particular, if $\delta$ is the square-free part of $\Delta$, then all
eigenvalues of $\Gamma_\chi$ are integral multiples of $\sqrt{\delta}$.
\end{corollary}

\begin{proof}
Proposition~\ref{prop:parity} shows that the connection set of $\Gamma_\chi$
is oriented.  Since $1\in C_\chi^+(\Delta)$, it generates the additive group
$\Z_\Delta$, so $\Gamma_\chi$ is connected.  The eigenvalue formula follows
from Corollary~\ref{cor:coprime-spectrum} with $m=1$, $X=\{1\}$, and
$\sigma(1)=1$.
The value $0$ occurs at the $\Delta-\varphi(\Delta)$ nonunits.  On the units,
$\chi$ takes each of the values $1$ and $-1$ exactly
$\varphi(\Delta)/2$ times, which gives the characteristic polynomial.
\end{proof}

For every positive square-free integer $\delta$, let $\chi$ be the odd
primitive quadratic character of conductor $\Delta(\delta)$.  Corollary~\ref{cor:primitive-spectrum} then supplies a connected oriented circulant
graph $\Gamma_\chi$ whose Hermitian eigenvalues are integral multiples of
$\sqrt{\delta}$.  Thus every positive square-free radicand occurs, although
the conductor and the radicand need not coincide.

\section{PST and MST for square-root spectra}\label{sec:transfer}

We first develop conditions for state transfer in an arbitrary oriented circulant graph
whose Hermitian eigenvalues lie in $\sqrt{D}\Z$, where $D>0$.  The symbol $D$
is reserved for this general spectral hypothesis; $\Delta$ continues to
denote the conductor of $\chi$ for the graphs in Definition~\ref{def:main-graph}.  These conditions reduce state transfer to congruences in
integer coefficients.  Combined with the spectral computations in the
preceding section, they yield the classification in
Section~\ref{sec:all-orders}.

Let $\Gamma$ be an oriented circulant graph on $\Z_n$ with Hermitian adjacency
matrix $A$.  Its transition matrix at $t\in\mathbb R$ is
\[
   U(t)=\exp(\ii tA).
\]
We use this sign convention throughout.  Replacing $U(t)$ by
$\exp(-\ii tA)$ replaces $U(t)=P_s$ by $U(-t)=P_{-s}$, but does not change
any existence or classification result.
For $a,b\in\Z_n$, the graph $\Gamma$ has perfect state transfer (PST) from
$a$ to $b$ at time $t$ if
$|U(t)_{b,a}|=1$.  The case $a=b$ is periodicity at $a$.  Unless stated
otherwise, the vertices in a PST pair are distinct.  If $n$ is even, PST
between $a$ and $a+n/2$ is called antipodal PST.  The graph $\Gamma$ has
pretty good state transfer (PGST) from
$a$ to $b$ if, for every $\varepsilon>0$, there is a time $t$ such that
$|U(t)_{b,a}|>1-\varepsilon$.  Since $A$ is circulant, its spectral
idempotents are the rank-one Fourier idempotents, and
\[
   U(t)_{b,a}
   =
   \frac1n\sum_{j=0}^{n-1}
      \exp(\ii t\lambda_j)\zeta_n^{j(b-a)}.
\]
Let $e_x$ denote the standard basis vector indexed by $x\in\Z_n$, and let
$P_s$ be the permutation matrix defined by $P_s e_x=e_{x+s}$ for
$s\in\Z_n$.  A subset $S\subseteq\Z_n$ supports multiple state transfer (MST)
in $\Gamma$ if PST occurs between every pair of vertices in $S$, possibly at
different times.  This agrees with the definition for mixed graphs in
\cite{SongLin2026Mixed}.  If $|S|=k$, then $S$ supports $k$-vertex MST.  The
classification below will show that every maximum MST set is an equally spaced
subset of $\Z_n$.

\begin{proposition}\label{prop:pst-translation}
Let $\Gamma$ be an oriented circulant graph on $\Z_n$.  If $\Gamma$ has PST
from $a$ to $a+s$ at time $t$, then
\[
   U(t)=P_s.
\]
Consequently, PST on an oriented circulant graph determines the entire
transition matrix.
\end{proposition}

\begin{proof}
The Fourier expression for the transition amplitude is an average of $n$
complex numbers of modulus one.  If its modulus is one, all of these summands
must be equal.  Hence PST forces the Fourier phases to align:
\[
   \exp(\ii t\lambda_j)\zeta_n^{js}=\gamma
   \qquad (0\le j<n)
\]
for a unit complex number $\gamma$.  Since $\lambda_0=0$, the equation at
$j=0$ gives $\gamma=1$.  Thus, on the $j$-th Fourier vector, $U(t)$ has
eigenvalue $\zeta_n^{-js}$, which is also the eigenvalue of $P_s$.  The common
Fourier eigenbasis now gives $U(t)=P_s$.
\end{proof}

The next lemma specializes the periodicity theorem of Godsil and Lato to
oriented circulant graphs.  In the real skew-adjacency convention,
\cite[Theorem~6.1]{GodsilLato2020} characterizes periodicity at a vertex of a
connected oriented graph by the condition that the eigenvalues in its support
belong to $\Z\sqrt{-\delta}$ for a positive square-free integer $\delta$.
The corresponding statement for root-of-unity-weighted Cayley graphs also
appears in \cite[Lemma~10.1]{GodsilZhang2024}.

\begin{lemma}\label{lem:periodic-square-root}
Let $\Gamma$ be a nonempty oriented circulant graph with Hermitian adjacency
matrix $A$.  If
\[
   \exp(\ii T A)=I
\]
for some $T>0$, then there is a positive square-free integer $\delta$ such
that every eigenvalue of $A$ belongs to $\sqrt{\delta}\Z$.
\end{lemma}

\begin{proof}
Write $\Gamma=\Cay(\Z_n,C)$, and let $H$ be the subgroup of $\Z_n$ generated
by $C$.  Its connected components are the cosets of $H$, and each component
is a connected oriented circulant graph.  Since $\exp(\ii T A)=I$ on every
component, it suffices to consider one of them.

Let $E_\theta$ be a nonzero spectral idempotent of this connected component.
Translation invariance makes the diagonal entries of $E_\theta$ constant,
while their sum is $\operatorname{tr}(E_\theta)>0$.  Each diagonal entry is
therefore positive, and every eigenvalue belongs to the eigenvalue support of
every vertex.  Theorem~6.1 of \cite{GodsilLato2020} consequently implies that
the eigenvalues of the real skew-adjacency matrix belong to
$\Z\sqrt{-\delta}$ for some positive square-free integer $\delta$.
Multiplication by $\ii$ converts that matrix to the Hermitian adjacency
matrix used here, and consequently its eigenvalues belong to
$\sqrt{\delta}\Z$.  All components have the same spectrum, which proves the
claim for $\Gamma$.
\end{proof}

We can now combine the translation property of PST, the preceding periodicity
lemma, and the cyclic form of \cite[Theorem~8.1]{GodsilZhang2024}.

\begin{corollary}\label{cor:pst-character-family}
Let $\Gamma=\Cay(\Z_n,C)$ be a nonempty oriented circulant graph with PST
between distinct vertices.  Then there are a positive
square-free integer $\delta$, the associated odd primitive quadratic character
$\chi$ of conductor $\Delta(\delta)$, a nonempty admissible divisor set $X$,
and a map $\sigma:X\to\{\pm1\}$ such that
\[
   \Gamma=\mathcal I_{n,\chi}(X,\sigma).
\]
\end{corollary}

\begin{proof}
By Proposition~\ref{prop:pst-translation}, PST at time $t$ has the form
$U(t)=P_s$ for some nonzero $s\in\Z_n$.  If $k$ is the order of $s$ in
$\Z_n$, then $U(kt)=I$.  Consequently, $U(k|t|)=I$, and $k|t|>0$ because
the transfer vertices are distinct.  Lemma~\ref{lem:periodic-square-root}
therefore gives a positive square-free integer $\delta$ for which the
Hermitian spectrum lies in $\sqrt{\delta}\Z$.  Applying
Theorem~\ref{thm:cyclic-classification} to this spectral condition gives the
displayed representation.
\end{proof}

Thus every oriented circulant graph with PST between distinct vertices
satisfies the square-root spectral hypothesis used in the remainder of the
paper.

\subsection{Congruence conditions for PST}

For a graph whose eigenvalues belong to $\sqrt D\Z$, Fourier phase alignment
turns PST into a system of linear congruences.  The next theorem solves that
system after dividing the integer eigenvalue coefficients by their common
divisor.  In particular, it determines the transfer vertices without first
prescribing a transfer time.

Let $\Gamma$ be a nonempty oriented circulant graph with eigenvalues
\[
    \lambda_j=\eta_j\sqrt{D},\qquad \eta_j\in\Z,
\]
and suppose that not all $\eta_j$ are zero.  Define
\[
   g=\gcd\{|\eta_j|:0\le j<n\},\qquad \widehat\eta_j=\eta_j/g.
\]
We shall repeatedly use the following consequence of B\'ezout's identity.
Since
\[
   \gcd\{\widehat\eta_j:0\le j<n\}=1,
\]
there are integers $b_0,\ldots,b_{n-1}$ such that
\begin{equation}\label{eq:bezout-normalized}
   \sum_{j=0}^{n-1}b_j\widehat\eta_j=1.
\end{equation}
Consequently, if $x\widehat\eta_j\in\Z$ for every $j$, then $x\in\Z$; and
if $q\widehat\eta_j\equiv0\pmod N$ for every $j$, then
$q\equiv0\pmod N$.  Both assertions follow by multiplying by the
coefficients in \eqref{eq:bezout-normalized} and summing.

\begin{theorem}\label{thm:transfer-structure}
Let $\Gamma$ be a nonempty oriented circulant graph on $\Z_n$ with
$\lambda_j=\eta_j\sqrt{D}$, where $D>0$, $\eta_j\in\Z$, and not all
$\eta_j$ are zero.  Put
\[
   g=\gcd\{|\eta_j|:0\le j<n\},
   \qquad
   \widehat\eta_j=\eta_j/g,
\]
and define
\[
   \mathcal T
   =
   \{0\}\cup
   \{\,s\in\Z_n\setminus\{0\}:
       \text{PST occurs from }0\text{ to }s\,\}
\]
and
\begin{equation}\label{eq:transfer-gcd}
   h=\gcd\!\left(
      n,\{j\widehat\eta_k-k\widehat\eta_j:0\le j,k<n\}
   \right).
\end{equation}
Then the following statements hold.

\begin{enumerate}[label=(\roman*)]
\item A vertex $s\in\Z_n$ belongs to $\mathcal T$ if and only if
there is an integer $r$ such that
\begin{equation}\label{eq:universal-pst}
   r\widehat\eta_j+sj\equiv0\pmod n
   \qquad (0\le j<n).
\end{equation}
The residue of $r$ modulo $n$ is unique.  If $r_s$ is any representative,
then all real times $t$ for which $U(t)=P_s$ are
\[
   t=\frac{2\pi(r_s+qn)}{gn\sqrt{D}},
   \qquad q\in\Z.
\]

\item The set $\mathcal T$ is the subgroup
\[
   \mathcal T=\frac{n}{h}\Z_n,
   \qquad |\mathcal T|=h.
\]
For any distinct vertices $a,b\in\Z_n$,
\begin{equation}\label{eq:all-pst-pairs}
   \text{PST occurs from $a$ to $b$}
   \quad\Longleftrightarrow\quad
   b-a\in\mathcal T.
\end{equation}
If $\sum_{j=0}^{n-1}a_j\widehat\eta_j=1$, then the residue in part~(i)
is given by
\begin{equation}\label{eq:r-from-bezout}
   r_s\equiv-s\sum_{j=0}^{n-1}ja_j\pmod n.
\end{equation}

\item The minimum positive period of $U(t)$ is
\[
   T_0=\frac{2\pi}{g\sqrt{D}}.
\]
There is a unique isomorphism
\[
   \rho:\mathcal T\longrightarrow\frac{n}{h}\Z_n
\]
such that
\[
   \rho(s)\widehat\eta_j+sj\equiv0\pmod n
   \qquad(0\le j<n).
\]
If $\rho(s_0)=n/h$, then
\begin{equation}\label{eq:transfer-cycle}
   U\!\left(\frac{qT_0}{h}\right)=P_{qs_0}
   \qquad(q\in\Z).
\end{equation}
These are all translation matrices attained by $U(t)$.  If $h>1$, the
minimum positive time for PST between distinct vertices is $T_0/h$.

\item PGST from $a$ to $b$ occurs if and only if PST from $a$ to $b$
occurs.  Every coset of $\mathcal T$ supports MST, every set supporting MST
is contained in one such coset, and the maximum number of vertices in an
MST set is $h$.
\end{enumerate}
\end{theorem}

\begin{proof}
\noindent\textbf{Part (i).}
Suppose first that $s\ne0$.  Since $\lambda_0=0$, PST from $0$ to $s$ at
time $t$ has phase $1$, and the circulant symmetry gives $U(t)=P_s$.
Conversely, $U(t)=P_s$ implies PST from every vertex $a$ to $a+s$.
For $s=0$, we have $U(0)=I=P_0$.  Hence $s\in\mathcal T$ if and only if
$U(t)=P_s$ for some $t\in\mathbb R$.  By the Fourier decomposition, this
matrix identity is equivalent to
\[
   \exp(\ii t\lambda_j)\zeta_n^{js}=1
   \qquad (0\le j<n).
\]
Put $\tau=t\sqrt{D}/(2\pi)$.  After multiplication of the phase congruence
by $n$, we obtain
\[
   n\tau g\,\widehat\eta_j+sj\in n\Z
   \qquad(0\le j<n).
\]
Thus $n\tau g\,\widehat\eta_j$ is an integer for every $j$.
Equation~\eqref{eq:bezout-normalized} implies that $r=n\tau g$ is an
integer, and \eqref{eq:universal-pst} follows.  Conversely, if
\eqref{eq:universal-pst} holds, then
\[
   \frac{t\lambda_j}{2\pi}+\frac{sj}{n}
   =
   \frac{r\widehat\eta_j+sj}{n}\in\Z
\]
at $t=2\pi r/(gn\sqrt{D})$.  Hence all Fourier phases agree and
$U(t)=P_s$.  If two residues $r_1$ and $r_2$ satisfy
\eqref{eq:universal-pst}, then
$(r_1-r_2)\widehat\eta_j\equiv0\pmod n$ for every $j$.
Multiplication by the coefficients in \eqref{eq:bezout-normalized} and
summation give $r_1\equiv r_2\pmod n$.  This proves part~(i), including the
description of all representatives and hence of all transfer times.

\smallskip
\noindent\textbf{Part (ii).}
If $s\in\mathcal T$, elimination of $r$
from the congruences with indices $j$ and $k$ gives
\[
   s(j\widehat\eta_k-k\widehat\eta_j)\equiv0\pmod n
   \qquad(0\le j,k<n).
\]
By B\'ezout's identity and \eqref{eq:transfer-gcd}, these congruences imply
$n\mid sh$, or equivalently $n/h\mid s$.  Conversely, suppose that
$n/h\mid s$ and choose integers $a_j$ with
$\sum_j a_j\widehat\eta_j=1$.  Define $r_s$ by
\eqref{eq:r-from-bezout}.  For every $k$,
\[
\begin{aligned}
   r_s\widehat\eta_k+sk
   &\equiv
   s\sum_{j=0}^{n-1}a_j
      (k\widehat\eta_j-j\widehat\eta_k)
      \pmod n\\
   &\equiv0\pmod n.
\end{aligned}
\]
Therefore \eqref{eq:universal-pst} holds, proving part~(ii).
The equivalence \eqref{eq:all-pst-pairs} follows because every circulant
translation is a graph automorphism: PST from $a$ to $b$ is equivalent to
PST from $0$ to $b-a$.

\smallskip
\noindent\textbf{Part (iii).}
At time $T_0=2\pi/(g\sqrt{D})$, all Fourier eigenvalues of $U(T_0)$ are
equal to $1$.  If $U(T)=I$ and $\alpha=T\sqrt{D}/(2\pi)$, then
$\alpha\eta_j\in\Z$ for every $j$.  Equation
\eqref{eq:bezout-normalized} gives $\alpha g\in\Z$, so $T_0$ is the
minimum positive period.  For $s\in\mathcal T$, let $\rho(s)$ be the unique
residue supplied by part~(i).  The defining congruences show that $\rho$ is
a homomorphism.  If $\rho(s)=0$, the congruence at $j=1$ gives $s=0$;
hence $\rho$ is injective.  Its image is a subgroup of $\Z_n$ of order
$h$ and is therefore $(n/h)\Z_n$.  Taking $\rho(s_0)=n/h$ now gives
\eqref{eq:transfer-cycle}.  The uniqueness in part~(i) shows that no other
translation matrices occur and also gives the minimum nonzero transfer
time.  This proves part~(iii).

\smallskip
\noindent\textbf{Part (iv).}
If PGST occurs, reduce a sequence of
approximating times modulo $T_0$.  A convergent subsequence exists, and
continuity of $U(t)$ gives PST at its limit.  The converse is immediate.
Finally, the difference of two vertices in a coset of $\mathcal T$ belongs
to $\mathcal T$, so each coset supports MST.  Conversely, translate an MST
set so that it contains $0$; all its vertices must then belong to
$\mathcal T$.  The maximum cardinality is therefore $|\mathcal T|=h$.
\end{proof}

\noindent
In view of part~(ii), we call $\mathcal T$ the transfer subgroup of
$\Gamma$.

\begin{remark}
For integral oriented circulant graphs, \cite[Lemma~3.4]{Song2024IOC}
reduces the test for PST to the condition that
\[
   v_2(\lambda_{j+1}-\lambda_j)=m
   \qquad (0\le j<n)
\]
for some $m\in\mathbb N$.  Equation~\eqref{eq:universal-pst} extends that
test to square-root spectra.  After the common factor $\sqrt D$ and the gcd
of the integer eigenvalue coefficients are removed, the congruence retains
the transfer vertex and determines both the transfer subgroup and every PST
time.  For the conductor-$4$ family, it specializes to the $2$-adic
condition in \cite[Lemma~3.4]{Song2024IOC}.
\end{remark}

\begin{corollary}\label{cor:component-transfer}
Let $\Gamma=\mathcal I_{n,\chi}(X,\sigma)$ be a nonempty oriented circulant
graph.  Let
$\kappa$, $n'$, and $\Gamma'$ be as in
Proposition~\ref{prop:connectedness}, and let $\mathcal T'$ be the transfer
subgroup of $\Gamma'$.  Then
for $s\in\Z_n\setminus\{0\}$, PST occurs in $\Gamma$ from $0$ to $s$ at time
$t$ if and only if $s=\kappa s'$ for some $s'\in\Z_{n'}$ and PST occurs in
$\Gamma'$ from $0$ to $s'$ at the same time.  If $\mathcal T$ is the
transfer subgroup of $\Gamma$, then
\[
   \mathcal T=\kappa\mathcal T'
   =\{\kappa s':s'\in\mathcal T'\}\subseteq\Z_n.
\]
Moreover, $\Gamma$ and $\Gamma'$ have the same positive periods and the same
maximum size of a set supporting MST.
\end{corollary}

\begin{proof}
Order the vertices of $\Gamma$ by their residue classes modulo $\kappa$.
Proposition~\ref{prop:connectedness} shows that $A(\Gamma)$ is then the direct
sum of $\kappa$ copies of $A(\Gamma')$.  The same is true of the transition
matrix at every time.  Hence no state transfer can occur between distinct
components.  Within each component, the identification
$c+\kappa y\mapsto y$ preserves the transition matrix, and translation by
$s=\kappa s'$ corresponds to translation by $s'$ in $\Gamma'$.  This proves
the equivalence and the formula for $\mathcal T$.  The assertions about
periodicity and MST follow from the same direct-sum decomposition.
\end{proof}

The following proposition gives a convenient test for a prescribed equally
spaced $k$-vertex MST.

\begin{proposition}\label{prop:shift}
Let $\Gamma$ be an oriented circulant graph on $\Z_n$ whose Hermitian
eigenvalues are written as
\[
    \lambda_j=\eta_j\sqrt{D},\qquad \eta_j\in\Z.
\]
Let $k\mid n$, put $s=n/k$, and let $\rho\in\Z$.  Then
\[
    U\!\left(\frac{2\pi\rho}{k\sqrt{D}}\right)=P_s
\]
if and only if
\begin{equation}\label{eq:shift-congruence}
   \rho\eta_j+j\equiv0\pmod k
   \qquad\text{for all }j\in\Z_n.
\end{equation}
\end{proposition}

\begin{proof}
The eigenvalue of $P_s$ on the $j$-th Fourier vector is
$\zeta_n^{-js}=\zeta_k^{-j}$.  At the time in the statement, the eigenvalue
of $U(t)$ on the same Fourier vector is
\[
   \exp\!\left(\frac{2\pi\ii\rho}{k}\eta_j\right).
\]
The two eigenvalues agree for every $j$ if and only if
$\rho\eta_j+j\equiv0\pmod k$, which is
\eqref{eq:shift-congruence}.
\end{proof}

\section{Complete classification}\label{sec:classification}

The following lemma records how the eigenvalue coefficients transform under
multiplication of the Fourier index by a unit.

\begin{lemma}\label{lem:unit-covariance}
Let $\Gamma=\mathcal I_{n,\chi}(X,\sigma)$ be an oriented circulant graph and write
$\lambda_j=\eta_j\sqrt{\Delta}$.  If $u\in\Z_n^\times$, then
\[
   \eta_{uj}=\chi(u)\eta_j
   \qquad (j\in\Z_n).
\]
\end{lemma}

\begin{proof}
In \eqref{eq:general-spectrum}, multiplication by $u$ preserves the condition
$Q_d\mid j$.  Moreover,
\[
   \chi(uj/Q_d)=\chi(u)\chi(j/Q_d),
   \qquad c_{M_d}(uj)=c_{M_d}(j),
\]
because $u$ is a unit modulo both $\Delta$ and $M_d$.  Each summand in
\eqref{eq:general-spectrum} therefore acquires the factor $\chi(u)$, which
proves the result.
\end{proof}

The covariance in Lemma~\ref{lem:unit-covariance} imposes a strong
restriction on the vertices that can be joined to $0$ by PST, even when
$n/\Delta$ is not coprime to $\Delta$.

\begin{lemma}\label{lem:annihilator-order}
Let $s\in\Z_n$ have order $k$.  For every integer $x$,
\[
   sx\equiv0\pmod n
   \quad\Longleftrightarrow\quad
   x\equiv0\pmod k.
\]
\end{lemma}

\begin{proof}
Put $d=\gcd(n,s)$.  Then $n=dk$ and $s=dc$ for some integer $c$ coprime
to $k$.  Hence $n\mid sx$ if and only if $k\mid cx$, which is equivalent
to $k\mid x$.
\end{proof}

\begin{lemma}\label{lem:primitive-conductor}
Let $\psi_1$ and $\psi_2$ be primitive Dirichlet characters of conductors
$f_1$ and $f_2$, respectively, and let $N$ be a common multiple of
$f_1$ and $f_2$.  If the characters induced by $\psi_1$ and $\psi_2$ on
$\Z_N^\times$ are equal, then $f_1=f_2$ and $\psi_1=\psi_2$.
\end{lemma}

\begin{proof}
Every Dirichlet character is induced by a unique primitive Dirichlet
character.  Applying this standard uniqueness result to the common
character modulo $N$ proves the assertion; see, for example,
\cite{IrelandRosen}.
\end{proof}

\begin{theorem}\label{thm:global-conductor}
Let $\chi$ be an odd primitive quadratic character of conductor $\Delta$,
and let $\Gamma=\mathcal I_{\Delta m,\chi}(X,\sigma)$ be a nonempty oriented
circulant graph.  If $\Gamma$ has PST between distinct vertices, then
\[
   \Delta\in\{3,4,8\}.
\]
More precisely, if PST occurs from $0$ to $s\ne0$, then the order of $s$ in
$\Z_{\Delta m}$ is $3$ when $\Delta=3$, belongs to $\{2,4\}$ when
$\Delta=4$, and is $2$ when $\Delta=8$.
\end{theorem}

\begin{proof}
Suppose that PST occurs from $0$ to $s\ne0$, and let
\[
   k=\frac{n}{\gcd(n,s)}
\]
be the order of $s$ in $\Z_n$, where $n=\Delta m$.  By
Theorem~\ref{thm:general-spectrum}, write
$\lambda_j=\eta_j\sqrt{\Delta}$, and set
$g=\gcd\{|\eta_j|:0\le j<n\}$ and $\widehat\eta_j=\eta_j/g$.  By Theorem~\ref{thm:transfer-structure}, there is an integer $r$ such that
\[
   r\widehat\eta_j+sj\equiv0\pmod n
   \qquad (j\in\Z_n).
\]
Apply the congruence at $uj$, use Lemma~\ref{lem:unit-covariance}, and
subtract $\chi(u)$ times the congruence at $j$.  At $j=1$ this yields
\[
   s\bigl(u-\chi(u)\bigr)\equiv0\pmod n.
\]
Lemma~\ref{lem:annihilator-order} therefore gives
\begin{equation}\label{eq:unit-order-restriction}
   u\equiv\chi(u)\pmod k
   \qquad (u\in\Z_n^\times).
\end{equation}

By Lemma~\ref{lem:unit-reduction}, the map
$\Z_n^\times\to\Z_\ell^\times$ is onto whenever $\ell\mid n$.  We now analyze
the possible prime divisors of $k$.  Suppose first that an odd prime $p$
divides $k$.  Surjectivity of the reduction map to $\Z_p^\times$, together with
\eqref{eq:unit-order-restriction}, shows that every nonzero residue modulo
$p$ is $1$ or $-1$.  Hence $p=3$.  If $9\mid k$, surjectivity of the
reduction map to $\Z_9^\times$ gives a unit $u\equiv2\pmod9$, whereas
\eqref{eq:unit-order-restriction} would give $2\equiv\chi(u)\equiv\pm1\pmod9$,
a contradiction.  Thus $v_3(k)\le1$.  If $3\mid k$, then
\eqref{eq:unit-order-restriction} says that the character induced by $\chi$
on $\Z_n^\times$ equals the character induced by the non-principal
character modulo $3$.  Lemma~\ref{lem:primitive-conductor} gives
$\Delta=3$.  Similarly, if $4\mid k$, then
$\chi(u)\equiv u\pmod4$ for every unit $u$.  The character induced by
$\chi$ on $\Z_n^\times$ is therefore the character induced by $\chi_4$, and
Lemma~\ref{lem:primitive-conductor} gives $\Delta=4$.  Finally, $8\nmid k$: a unit
$u\equiv3\pmod8$ cannot satisfy $u\equiv\pm1\pmod8$.  Also, $12\nmid k$,
because $3\mid k$ and $4\mid k$ would force both $\Delta=3$ and $\Delta=4$.
Since $s\ne0$, we have $k>1$.  The preceding restrictions therefore leave only
\[
   k\in\{2,3,4,6\}.
\]
The cases $k=3$ and $k=4$ give the orders listed in the theorem,
while $k=6$ can occur only when $\Delta=3$.

It remains to exclude $k=2$ when $\Delta\notin\{4,8\}$ and $k=6$ when
$\Delta=3$.

\smallskip
\noindent\textbf{Case 1.} $k=2$.
Then $n$ is even and $s=n/2$.
By Proposition~\ref{prop:delta}, the only conductors of odd primitive
quadratic characters whose prime divisors are all equal to $2$ are $4$ and
$8$.  Hence the fundamental discriminant $-\Delta$ has an odd prime divisor
$p$.  Put $a=v_p(m)$ and $j=p^{a+1}$.  We have $0<j<n$: the only possible
equality $j=n$ would require $\Delta=p$ and $m=p^a$, contrary to the parity of
$n$.  For every term in \eqref{eq:general-spectrum}, we have $q_d=m/d$ and
hence $v_p(Q_d)\le a$.  Therefore either $Q_d\nmid j$, or
$p\mid j/Q_d$; in the latter case $\chi(j/Q_d)=0$.  Thus $\eta_j=0$.
Since $j$ is odd, the congruence at this index would require
\[
   0\equiv sj=\frac n2j\not\equiv0\pmod n,
\]
a contradiction.

\smallskip
\noindent\textbf{Case 2.} $k=6$ and $\Delta=3$.
Again $n$ is even.
With $a=v_3(m)$, the preceding calculation at the Fourier index
$3^{a+1}$ gives $\eta_{3^{a+1}}=0$ and $3^{a+1}<n$.  Writing
$s=(n/6)c$ with $(c,6)=1$, the congruence at this index would imply
$6\mid c3^{a+1}$, which is impossible.  This contradiction leaves only the
conductors and orders in the statement.
\end{proof}

Theorem~4.1 of Chan et al.\
\cite{ChanPantangiRazafimahatratraSin2026} already implies that an oriented
normal Cayley graph on a solvable group has no set supporting pairwise PST
with more than four vertices.  Since cyclic groups are solvable, the next
corollary records this bound in the present setting and refines it according
to the conductor.

\begin{corollary}\label{cor:global-mst-bound}
Let $\Gamma$ be a nonempty oriented circulant graph.  Every set supporting
pairwise PST in $\Gamma$ has at most four vertices.  If $\Gamma$ has PST
between distinct vertices, let $\Delta$ be the conductor in the
representation supplied by Corollary~\ref{cor:pst-character-family}.  A
four-vertex MST set can occur only when $\Delta=4$.  If $\Delta=3$, every MST
set has at most three vertices; if $\Delta=8$, every MST set has at most two
vertices.
\end{corollary}

\begin{proof}
If $\Gamma$ has no PST between distinct vertices, the assertion is immediate.
Otherwise, Corollary~\ref{cor:pst-character-family} represents $\Gamma$ as
$\mathcal I_{\Delta m,\chi}(X,\sigma)$.  Let $\mathcal T$ be its transfer
subgroup.
Theorem~\ref{thm:global-conductor} lists the possible orders of its nonzero
elements.  Since every subgroup of $\Z_n$ is cyclic, its order is at most
$4$, with the sharper bounds for $\Delta=3$ and $\Delta=8$.
\end{proof}

\begin{remark}
The general four-vertex bound in
\cite[Theorem~4.1]{ChanPantangiRazafimahatratraSin2026} does not determine
which cyclic connection sets attain orders $2$, $3$, or $4$.
Corollary~\ref{cor:global-mst-bound} identifies the conductor required for
each of these orders, and Theorem~\ref{thm:all-orders-classification} below
determines exactly when they occur.
\end{remark}

\subsection{Classification at arbitrary orders}\label{sec:all-orders}

We now complete the classification at arbitrary orders.  By Corollary~\ref{cor:pst-character-family}, which combines
\cite[Theorem~6.1]{GodsilLato2020} with
\cite[Theorem~8.1]{GodsilZhang2024}, every nonempty oriented circulant graph
with PST has the form $\mathcal I_{n,\chi}(X,\sigma)$.  Theorem~\ref{thm:global-conductor} restricts its conductor to $3$, $4$, or $8$.
It therefore remains to determine the admissible divisors in $X$ at arbitrary
multiples of these conductors.  The integer $m$ below need neither be coprime
to the conductor nor be a prime power.  Since $n=\Delta m$, we have
\[
   \Delta\mid n/d\quad\Longleftrightarrow\quad d\mid m.
\]
Thus the admissible gcd-class indices are precisely the divisors of $m$.  For
$\Delta\in\{3,4,8\}$, define
\[
   p=
   \begin{cases}
      3,&\Delta=3,\\
      2,&\Delta\in\{4,8\},
   \end{cases}
   \qquad
   m=p^a u,\quad (p,u)=1.
\]
For $0\le r\le a$, define
\[
   X_r=X\cap\{p^{a-r}e:e\mid u\}.
\]
Thus $X_0$ consists of the selected divisors having the same $p$-adic
valuation as $m$.  The remaining sets $X_r$ are ordered by decreasing
$p$-adic valuation.

\begin{theorem}\label{thm:all-orders-classification}
Let $\Delta\in\{3,4,8\}$, let $\chi$ be the odd primitive quadratic character
of conductor $\Delta$, and let
$\Gamma=\mathcal I_{\Delta m,\chi}(X,\sigma)$ be a nonempty oriented
circulant graph.  Put $n=\Delta m$, and let $\mathcal T$ be its transfer
subgroup.  Write
\[
   \lambda_j=\eta_j\sqrt{\Delta},\qquad
   g=\gcd\{|\eta_j|:0\le j<n\},\qquad
   \widehat\eta_j=\eta_j/g.
\]

If $\Delta=3$, then
\[
   \mathcal T=
   \begin{cases}
      m\Z_{3m},&X_0=\{m\},\\
      \{0\},&X_0\ne\{m\}.
   \end{cases}
\]
If $X_0=\{m\}$, then
\[
   U\!\left(\frac{2\pi}{3\sqrt3}\right)=P_{\sigma(m)m}.
\]

If $\Delta=4$, then
\[
\mathcal T=
\begin{cases}
\{0\},&X_0\ne\{m\},\\
\{0,2m\},&a=0\text{ and }X_0=\{m\},\\
\{0,2m\},&a\ge1,\ X_0=\{m\},\ X_1\ne\{m/2\},\\
m\Z_{4m},&a\ge1,\ X_0=\{m\},\ X_1=\{m/2\}.
\end{cases}
\]
In the second and third cases, $U(\pi/2)=P_{2m}$.  In the fourth case,
\[
   U(\pi/4)=P_{\sigma(m)m}.
\]

If $\Delta=8$, then
\[
   \mathcal T=
   \begin{cases}
      \{0,4m\},&X_0=\{m\},\\
      \{0\},&X_0\ne\{m\},
   \end{cases}
\]
and, in the first case, $U(\pi/\sqrt8)=P_{4m}$.  In every case above with
$\mathcal T\ne\{0\}$, we have $g=1$.
\end{theorem}

\begin{proof}
We begin by separating the coefficient sequence according to the $p$-adic
valuation of the Fourier index.  Let $1\le j<n$ and suppose that
$v_p(j)=r\le a$.  For a divisor $d=p^{a-s}e$, where $e\mid u$, the parameter
$Q_d$ in Theorem~\ref{thm:general-spectrum} is $p^s$.  If $s>r$, then
$Q_d\nmid j$; if $s<r$, then $\chi(j/Q_d)=0$.  Thus only the divisors in
$X_r$ contribute, and Theorem~\ref{thm:general-spectrum} gives
\begin{equation}\label{eq:mixed-block-spectrum}
   \eta_j=-p^r\chi(j/p^r)F_r(j),
   \qquad
   F_r(j)=
   \sum_{\substack{e\mid u\\p^{a-r}e\in X}}
      \sigma(p^{a-r}e)\chi(u/e)c_{u/e}(j).
\end{equation}
Also, $\eta_0=0$, and $\eta_j=0$ when $v_p(j)>a$.  Formula
\eqref{eq:mixed-block-spectrum} will be used in each of the three conductor
cases.

\smallskip
\noindent\textbf{Case 1.} $\Delta=3$.
Suppose that $\mathcal T\ne\{0\}$.  By
Theorem~\ref{thm:global-conductor}, every nonzero element of $\mathcal T$ has order
$3$.  Since $\mathcal T$ is cyclic, it is the unique subgroup of order $3$
in $\Z_{3m}$; hence $m\in\mathcal T$.  Theorem~\ref{thm:transfer-structure} gives an integer $R$ such
that
\[
   R\widehat\eta_j+mj\equiv0\pmod{3m}
   \qquad(0\le j<3m).
\]
Reduction modulo $m$ gives $R\widehat\eta_j\equiv0\pmod m$ for every $j$.
Multiplying by the coefficients in \eqref{eq:bezout-normalized} and summing
shows that $m\mid R$.  Write $R=m\rho$.  We obtain
\begin{equation}\label{eq:mixed-three-congruence}
   \rho\widehat\eta_j+j\equiv0\pmod3.
\end{equation}
At $j=1$, the congruence \eqref{eq:mixed-three-congruence} forces
$\widehat\eta_1\not\equiv0\pmod3$; equation~\eqref{eq:mixed-block-spectrum}
therefore shows that $X_0$ is nonempty.  We next show that the normalizing
factor $g$ is not divisible by $3$.  If $3\mid g$, then
$F_0(j)\equiv0\pmod3$ for every $j$ coprime to $3$.  Every residue modulo
$u$ has such a representative, so Lemma~\ref{lem:ramanujan-independent}
would force all coefficients in $F_0$ to vanish modulo $3$, a contradiction.
Thus $3\nmid g$.  Multiplying \eqref{eq:mixed-three-congruence} by $g$ and
using $\chi_3(j)\equiv j\pmod3$ shows that
$F_0(j)\equiv c\pmod3$ for a fixed $c\in\mathbb F_3^\times$.  Write
\[
   F_0(j)=\sum_{q\mid u}a_qc_q(j),
\]
where the coefficients are now regarded in $\mathbb F_3$ and $a_q$ is
nonzero precisely when $m/q\in X_0$.  Then
$F_0-cc_1$ vanishes modulo $3$.  Lemma~\ref{lem:ramanujan-independent}
gives $a_q=0$ for $q>1$ and $a_1=c\ne0$.  Hence $X_0=\{m\}$.
Conversely, assume that $X_0=\{m\}$.  Then $g=1$ and
\[
   \eta_j=-\sigma(m)\chi_3(j)\quad(3\nmid j),
   \qquad 3\mid\eta_j\quad(3\mid j).
\]
Since $\chi_3(j)\equiv j\pmod3$, these relations give
\[
   \eta_j+\sigma(m)j\equiv0\pmod3
   \qquad(j\in\Z_{3m}).
\]
Comparison of the Fourier eigenvalues gives
$U(2\pi/(3\sqrt3))=P_{\sigma(m)m}$.
Theorem~\ref{thm:global-conductor} excludes PST from $0$ to any other vertex.

\smallskip
\noindent\textbf{Case 2.} $\Delta\in\{4,8\}$ and antipodal transfer.
Suppose that $\Delta\in\{4,8\}$ and
$\mathcal T\ne\{0\}$.  Theorem~\ref{thm:global-conductor} implies that
$n/2\in\mathcal T$: it is the unique element of order $2$, and it is the
double of every element of order $4$.  Applying
Theorem~\ref{thm:transfer-structure} with
$s=n/2$ gives an integer $R$ such that
\[
   R\widehat\eta_j+\frac n2j\equiv0\pmod n
   \qquad(0\le j<n).
\]
Reduction modulo $n/2$, followed by
\eqref{eq:bezout-normalized}, shows that $n/2\mid R$.  Writing
$R=(n/2)\rho$ and dividing the congruence by $n/2$ gives
\[
   \rho\widehat\eta_j+j\equiv0\pmod2
\]
for some integer $\rho$.  At the index $j=1$, this congruence forces
$\widehat\eta_1$ to be odd, so equation~\eqref{eq:mixed-block-spectrum}
shows that $X_0$ is nonempty.  If
$2\mid g$, Lemma~\ref{lem:ramanujan-independent}, applied to $F_0$ modulo
$2$, gives a contradiction.  Hence $g$ is odd, and the last congruence is
equivalent to $\eta_j\equiv j\pmod2$.  Equation
\eqref{eq:mixed-block-spectrum} now gives $F_0(j)\equiv1\pmod2$ at every odd
index.  Write $F_0(j)=\sum_{q\mid u}a_qc_q(j)$, where $a_q$ is nonzero
precisely when $m/q\in X_0$, with the coefficients regarded in
$\mathbb F_2$.  Applying
Lemma~\ref{lem:ramanujan-independent} to $F_0-c_1$ gives
$a_q=0$ for $q>1$ and $a_1=1$.  Thus the parity condition is equivalent to
$X_0=\{m\}$.  Under this condition $g=1$, and
Proposition~\ref{prop:shift} gives the stated antipodal times.
This establishes the classification for $\Delta=8$ and the first two
alternatives for $\Delta=4$.

\smallskip
\noindent\textbf{Case 3.} $\Delta=4$ and transfer to a vertex of order $4$.
It remains to determine when the oriented circulant graph
$\Gamma=\mathcal I_{4m,\chi_4}(X,\sigma)$ has PST from $0$ to a vertex of
order $4$.  Suppose first that $\mathcal T$ contains such a vertex.  Then it
contains the unique subgroup $m\Z_{4m}$ of order $4$, and in particular
$m\in\mathcal T$.  The preceding argument gives $X_0=\{m\}$ and $g=1$.
Applying
Theorem~\ref{thm:transfer-structure} with $s=n/4=m$ gives an integer $R$
satisfying the
corresponding congruences modulo $n$.  Reduction modulo $n/4$, followed by
\eqref{eq:bezout-normalized}, shows that $n/4\mid R$.  Since $g=1$, writing
$R=(n/4)\rho$ and dividing by $n/4$ gives
\[
   \rho\eta_j+j\equiv0\pmod4.
\]
At odd indices this fixes $\rho\equiv\sigma(m)\pmod4$.  If $a=0$, the index
$j=2$ rules out the congruence.  If $a\ge1$, put $j=2w$ with $w$ odd and use
\eqref{eq:mixed-block-spectrum}.  After division by $2$, the congruence is
equivalent to $F_1(2w)\equiv1\pmod2$ for all such $w$.  Since $u$ is odd,
these arguments cover every residue modulo $u$.  Lemma~\ref{lem:ramanujan-independent} shows that this occurs exactly when
$X_1=\{m/2\}$: in the expansion
$F_1(2w)=\sum_{q\mid u}a_qc_q(2w)$ over $\mathbb F_2$, only the coefficient
with $q=1$ can be nonzero, and this coefficient corresponds to the divisor
$m/2$.  Thus $a\ge1$ and $X_1=\{m/2\}$ are necessary.

Conversely, suppose that $a\ge1$, $X_0=\{m\}$, and $X_1=\{m/2\}$.  Then
$g=1$.  We verify the required congruence according to the $2$-adic
valuation of $j$.  If $j$ is odd, only the divisor $m\in X_0$ contributes
to \eqref{eq:mixed-block-spectrum}, and hence
\[
   \eta_j=-\sigma(m)\chi_4(j)
   \equiv-\sigma(m)j\pmod4.
\]
If $v_2(j)=1$, write $j=2w$ with $w$ odd.  Only the divisor
$m/2\in X_1$ contributes, so
\[
   \eta_j=-2\sigma(m/2)\chi_4(w).
\]
Both $\sigma(m/2)\chi_4(w)$ and $\sigma(m)w$ are odd.  Therefore
\[
   \eta_j+\sigma(m)j
   =2\bigl(-\sigma(m/2)\chi_4(w)+\sigma(m)w\bigr)
   \equiv0\pmod4.
\]
Finally, if $v_2(j)\ge2$, then $4\mid j$, and
\eqref{eq:mixed-block-spectrum} shows that $4\mid\eta_j$ as well.  Thus
\[
   \eta_j+\sigma(m)j\equiv0\pmod4
   \qquad(j\in\Z_{4m}).
\]
Comparison of the Fourier eigenvalues gives
$U(\pi/4)=P_{\sigma(m)m}$.
Theorem~\ref{thm:global-conductor} excludes PST from $0$ to vertices outside
the displayed subgroup.  This completes the three conductor cases.
\end{proof}

\begin{proof}[Proof of Theorem~\ref{thm:introduction-classification}]
Suppose that $\Gamma$ has PST between distinct vertices.  By
Corollary~\ref{cor:pst-character-family}, it has a representation
$\Gamma=\mathcal I_{n,\chi}(X,\sigma)$.  Theorem~\ref{thm:global-conductor}
then gives $\Delta\in\{3,4,8\}$; the primitive character $\chi$ is
$\chi_\Delta$.  Theorem~\ref{thm:all-orders-classification} gives precisely
the cases for which $\mathcal T\ne\{0\}$.  The converse follows from
Theorem~\ref{thm:all-orders-classification}.
\end{proof}

For the graphs in Theorem~\ref{thm:all-orders-classification},
Theorem~\ref{thm:transfer-structure} shows that PGST between distinct
vertices occurs exactly in the PST cases.  Thus approximation introduces no
additional cases.
Since $g=1$ whenever $\mathcal T\ne\{0\}$, part~(iii) of the same theorem
shows that the displayed times in
Theorem~\ref{thm:all-orders-classification} are the minimum positive PST
times.

Table~\ref{tab:all-orders-classification} summarizes the cases in
Theorem~\ref{thm:all-orders-classification} for which
$\mathcal T\ne\{0\}$.  In the fourth column, $t_{\min}$ is the minimum
positive time at which PST occurs between distinct vertices.

\begin{table}[htbp]
\centering
\small
\renewcommand{\arraystretch}{1.18}
\caption{PST cases in Theorem~\ref{thm:all-orders-classification}}
\label{tab:all-orders-classification}
\vspace{0.5em}
\begin{tabular}{ccccc}
\toprule
$\Delta$ & condition & $\mathcal T$ & $U(t_{\min})$ and $t_{\min}$ &
largest MST size \\
\midrule
$3$ & $X_0=\{m\}$ & $m\Z_{3m}$ &
$P_{\sigma(m)m},\ \frac{2\pi}{3\sqrt3}$ & $3$ \\[3pt]
$4$ & $\begin{gathered}a=0,\ X_0=\{m\},\ \text{or}\\
   a\ge1,\ X_0=\{m\},\ X_1\ne\{m/2\}\end{gathered}$ & $\{0,2m\}$ &
$P_{2m},\ \frac{\pi}{2}$ & $2$ \\[3pt]
$4$ & $a\ge1$, $X_0=\{m\}$, $X_1=\{m/2\}$ & $m\Z_{4m}$ &
$P_{\sigma(m)m},\ \frac{\pi}{4}$ & $4$ \\[3pt]
$8$ & $X_0=\{m\}$ & $\{0,4m\}$ &
$P_{4m},\ \frac{\pi}{\sqrt8}$ & $2$ \\
\bottomrule
\end{tabular}
\end{table}

For every graph in the table, nonzero transfer requires $X_0=\{m\}$; among
the selected divisors having the same $p$-adic valuation as $m$, only $m$ may
occur.  At conductor $4$, the additional condition $X_1=\{m/2\}$ changes the
cardinality of $\mathcal T$ from $2$ to $4$.

\section{Consequences, examples, and enumeration}

\subsection{The coprime case}

When $(\Delta,m)=1$, the conductor restricts both the order in $\Z_n$ and the
possible vertices joined to $0$ by PST.

\begin{theorem}\label{thm:coprime-pst-pairs}
Let $\chi$ be an odd primitive quadratic character of conductor $\Delta$,
and let $\Gamma=\mathcal I_{\Delta m,\chi}(X,\sigma)$ be an oriented
circulant graph with $(\Delta,m)=1$.  If $\Gamma$ has PST from $a$ to
$b$, then
\[
   b-a\equiv0\pmod m.
\]
\end{theorem}

\begin{proof}
The assertion is immediate if $m=1$.  Assume $m>1$, so that
$j=\Delta$ is a nonzero Fourier index modulo $n=\Delta m$.
From \eqref{eq:eta}, $\eta_j=0$ whenever $(j,\Delta)>1$.  In particular
\[
   \eta_0=\eta_\Delta=0.
\]
We compare these two zero eigenvalues through the Fourier phase condition.
PST from $a$ to $b$ at time $t$ is equivalent to the existence of a unit
complex number $\gamma$ such that
\[
   \exp(\ii t\lambda_j)\zeta_n^{j(b-a)}=\gamma
   \qquad (0\le j<n).
\]
Taking $j=0$ and $j=\Delta$, and using
$\lambda_0=\lambda_\Delta=0$, gives
\[
   \zeta_n^{\Delta(b-a)}=1.
\]
Since $n=\Delta m$, this is equivalent to $m\mid (b-a)$.
\end{proof}

\begin{remark}
For an oriented circulant graph $\Gamma=\mathcal I_{\Delta m,\chi}(X,\sigma)$ with
$(\Delta,m)=1$, state transfer can occur only between vertices in the same
residue class modulo $m$.  Equivalently, transfer is confined to the
$\Delta$-vertex cosets of $m\Z_n$.
\end{remark}

\begin{corollary}\label{cor:coprime-case}
Let $\Gamma=\mathcal I_{\Delta m,\chi}(X,\sigma)$ be a nonempty oriented
circulant graph with $(\Delta,m)=1$, and write
\[
   \lambda_j=\eta_j\sqrt{\Delta},\qquad
   g=\gcd\{|\eta_j|:0\le j<n\},\qquad
   \widehat\eta_j=\eta_j/g.
\]
Let $\mathcal T$ be the transfer subgroup of $\Gamma$.
The following statements hold.
\begin{enumerate}[label=\textup{(\roman*)}]
\item PST from $a$ to $a+s$ occurs if and only if there are integers
      $\ell$ and $r$ such that
      \[
         s\equiv\ell m\pmod n,\qquad
         r\widehat\eta_j+\ell j\equiv0\pmod\Delta
         \quad(0\le j<n).
      \]
      The integer $r$ may be chosen positive, and then
      $2\pi r/(g\Delta\sqrt{\Delta})$ is a PST time.
\item PST between distinct vertices occurs if and only if
      $\Delta\in\{3,4,8\}$ and $X=\{m\}$.  In these cases,
      $\mathcal T=m\Z_{3m}$ for $\Delta=3$, whereas
      $\mathcal T=\{0,n/2\}$ for $\Delta\in\{4,8\}$.
\item If $k\mid n$ and $U(t)=P_{n/k}$ for some $t$, then $k\mid\Delta$.
\end{enumerate}
\end{corollary}

\begin{proof}
\noindent\textbf{Part (i).}
Theorem~\ref{thm:coprime-pst-pairs} gives
$s\equiv\ell m\pmod n$.  Put $\tau=t\sqrt{\Delta}/(2\pi)$.  The Fourier
phase equations are
\[
   \tau\eta_j+\frac{\ell j}{\Delta}\in\Z
   \qquad(0\le j<n).
\]
Since the integers $\widehat\eta_j$ have gcd $1$, these equations imply that
$r=\Delta\tau g$ is an integer and give the displayed congruence.  The same
calculation in reverse proves sufficiency and the formula for a PST time.
Adding a multiple of $\Delta$ to $r$ permits $r$ to be chosen
positive.

\smallskip
\noindent\textbf{Part (ii).}
The exponent $a$ in
Theorem~\ref{thm:all-orders-classification} is zero, so $X_0=X$.  That theorem
gives exactly the three stated conductors, the condition $X=\{m\}$, and the
listed values of $\mathcal T$.

\smallskip
\noindent\textbf{Part (iii).}
If $U(t)=P_{n/k}$, then
Theorem~\ref{thm:coprime-pst-pairs} gives
$m\mid n/k=\Delta m/k$, which is equivalent to $k\mid\Delta$.
\end{proof}

The divisibility assertion in part~(iii) uses $(\Delta,m)=1$.  Without this
assumption, Theorem~\ref{thm:global-conductor} gives the appropriate
conductor-specific restriction.  A field-degree argument alone does not give
either conclusion, since the roots of unity $\exp(\ii t\lambda_j)$ also
depend on the transfer time.

\subsection{Power-of-two specializations}

For conductor $4$, let $a\ge3$,
put $n=2^a$, and choose
\[
   \varepsilon=(\varepsilon_0,\ldots,\varepsilon_{a-2})
   \in\{\pm1\}^{a-1}.
\]
For a nonzero $x\in\Z_n$, let $v_2(x)$ be the $2$-adic valuation of the unique
representative of $x$ in $\{1,\ldots,n-1\}$, and define
\[
   C^{(4)}_\varepsilon
   =\left\{x\in\Z_n\setminus\{0,n/2\}:
      \chi_4\!\left(\frac{x}{2^{v_2(x)}}\right)
      =\varepsilon_{v_2(x)}\right\}.
\]
Let $\Gamma^{(4)}_\varepsilon=\Cay(\Z_n,C^{(4)}_\varepsilon)$ be the
oriented circulant graph determined by $C^{(4)}_\varepsilon$.

\begin{theorem}\label{thm:power-two-four}
The oriented circulant graph $\Gamma^{(4)}_\varepsilon$ is connected.  Its
Hermitian eigenvalues are $\lambda_j=2\eta_j$, where $\eta_0=0$ and, for
$1\le j<n$, writing $r=v_2(j)$ and $u=j/2^r$,
\[
   \eta_j=
   \begin{cases}
   -\varepsilon_{a-2-r}2^r\chi_4(u), & 0\le r\le a-2,\\
   0, & r=a-1.
   \end{cases}
\]
Moreover,
\[
   \mathcal T=2^{a-2}\Z_{2^a},
   \qquad
   U_{\Gamma^{(4)}_\varepsilon}(\pi/4)
   =P_{\varepsilon_{a-2}2^{a-2}}.
\]
Consequently, every coset of $2^{a-2}\Z_{2^a}$ supports 4-vertex MST in
$\Gamma^{(4)}_\varepsilon$.
\end{theorem}

\begin{proof}
We first verify the graph-theoretic assertions.  Each nonzero difference other
than $n/2$ belongs to exactly one of the two sets determined by $\chi_4$, so
the connection set is oriented.  It contains an odd integer and hence
generates $\Z_{2^a}$.  For the gcd-class indexed by $d=2^s$, the parameter
$Q$ in Theorem~\ref{thm:general-factorization} is $2^{a-s-2}$.  Hence, when
$v_2(j)=r\le a-2$, only the class with $s=a-2-r$ contributes.  This gives
the displayed eigenvalue formula; no class contributes when $r=a-1$.

We next determine the transfer translation.  For odd $j$, we have
$\chi_4(j)\equiv j\pmod4$, and hence
\[
   \eta_j+\varepsilon_{a-2}j\equiv0\pmod4.
\]
The same congruence holds when $v_2(j)=1$, because both terms are congruent
to $2$ modulo $4$, and it holds when $4\mid j$ as well.  Comparing Fourier
eigenvalues gives the stated transition matrix.  Since $|\eta_1|=1$,
normalization in Theorem~\ref{thm:transfer-structure} does not change the
coefficients $\eta_j$.
Finally, $\eta_1=\eta_5$.  The two congruences in
Theorem~\ref{thm:transfer-structure} with indices $1$ and $5$ imply
$4s\equiv0\pmod{2^a}$ for every $s\in\mathcal T$.  Therefore
$\mathcal T=2^{a-2}\Z_{2^a}$.
\end{proof}

The square-free radicand $2$ corresponds to conductor $8$.  Let $a\ge3$,
put $n=2^a$, and choose
\[
   \varepsilon=(\varepsilon_0,\ldots,\varepsilon_{a-3})
   \in\{\pm1\}^{a-2}.
\]
Define
\[
   C^{(8)}_\varepsilon
   =\left\{x\in\Z_n\setminus\{0\}:v_2(x)\le a-3,
      \ \chi_8\!\left(\frac{x}{2^{v_2(x)}}\right)
      =\varepsilon_{v_2(x)}\right\}.
\]
Let $\Gamma^{(8)}_\varepsilon=\Cay(\Z_n,C^{(8)}_\varepsilon)$ be the
oriented circulant graph determined by $C^{(8)}_\varepsilon$.

\begin{corollary}\label{cor:power-two-eight}
The oriented circulant graph $\Gamma^{(8)}_\varepsilon$ is connected.  Its
Hermitian eigenvalues are $\lambda_j=\eta_j\sqrt8$, where $\eta_0=0$ and,
for $1\le j<n$, writing $r=v_2(j)$ and $u=j/2^r$,
\[
   \eta_j=
   \begin{cases}
   -\varepsilon_{a-3-r}2^r\chi_8(u), & 0\le r\le a-3,\\
   0, & a-2\le r\le a-1.
   \end{cases}
\]
For every choice of $\varepsilon$,
\[
   U_{\Gamma^{(8)}_\varepsilon}(\pi/\sqrt8)=P_{2^{a-1}},
   \qquad
   \mathcal T=\{0,2^{a-1}\}.
\]
Consequently, $\Gamma^{(8)}_\varepsilon$ has antipodal PST and has no
$k$-vertex MST for $k\ge3$.
\end{corollary}

\begin{proof}
The connection set contains a unit and is therefore connected.  For the
gcd-class indexed by $d=2^s$, Theorem~\ref{thm:general-factorization} has
$Q=2^{a-s-3}$.  Hence only $s=a-3-v_2(j)$ contributes, which gives the
displayed spectrum.  Since all admissible classes are selected, the
conductor-$8$ part of Theorem~\ref{thm:all-orders-classification} gives the
stated transition matrix and the value of $\mathcal T$.
\end{proof}

Theorem~\ref{thm:all-orders-classification} also covers power-of-two graphs
in which some admissible gcd-classes are omitted.  We next record the
prime-power specialization at odd conductors.

\subsection{Odd prime-power specializations}

We apply Theorem~\ref{thm:general-spectrum} to a connected family at odd
prime-power orders.  Let $p\equiv3\pmod4$ be prime, let $\chi_p$ be the
quadratic character modulo $p$, and put $n=p^a$, where $a\ge1$.  For
$x\in\Z_{p^a}\setminus\{0\}$, let $v_p(x)$
be the $p$-adic valuation of the unique representative of $x$ in
$\{1,\ldots,p^a-1\}$.  Thus $0\le v_p(x)<a$.  Fix
\[
   \varepsilon=(\varepsilon_0,\varepsilon_1,\ldots,\varepsilon_{a-1})
   \in\{\pm1\}^a,
\]
and define
\[
   C_{p,\varepsilon}
   =
   \left\{x\in\Z_{p^a}\setminus\{0\}:
       \chi_p\!\left(x/p^{v_p(x)}\right)
       =\varepsilon_{v_p(x)}\right\}.
\]
Let $\Gamma_{p,\varepsilon}=\Cay(\Z_{p^a},C_{p,\varepsilon})$ be the
oriented circulant graph with connection set $C_{p,\varepsilon}$.

\begin{theorem}\label{thm:ramified-family}
The set $C_{p,\varepsilon}$ is an oriented connection set,
$\Gamma_{p,\varepsilon}$ is connected, and the Hermitian eigenvalues of
$\Gamma_{p,\varepsilon}$ are
\[
   \lambda_0=0,
\]
and, for $1\le j<p^a$, writing $r=v_p(j)$ and $u=j/p^r$,
\[
   \lambda_j
   =
   -\varepsilon_{a-1-r}p^r\chi_p(u)\sqrt{p}.
\]
\end{theorem}

\begin{proof}
Since $\chi_p(-1)=-1$, the set $C_{p,\varepsilon}$ chooses exactly one
element from each pair $\{x,-x\}$, so it is an oriented connection set.  It
contains a unit, and every unit generates the additive group $\Z_{p^a}$.

It remains to calculate the eigenvalue at a fixed Fourier index.  The part
with $v_p(x)=s$ is indexed by $d=p^s$.  In Theorem~\ref{thm:general-factorization}, its parameters are
\[
   Q=p^{a-s-1},\qquad M=1.
\]
Its contribution at the index $j$ is nonzero precisely when
$v_p(j)=a-s-1$.  Thus the unique nonzero contribution is the part with
$s=a-r-1$.  Substitution gives the displayed eigenvalue formula.
\end{proof}

\begin{corollary}\label{cor:ramified-family-mst}
Let $\Gamma=\mathcal I_{p^a,\chi_p}(X,\sigma)$ be an oriented circulant
graph.  For $0\le s<a$, define $\beta_s=0$ if $p^s\notin X$ and
$\beta_s=\sigma(p^s)$ otherwise.  If $p=3$ and $\beta_{a-1}\ne0$, then
\[
   \mathcal T=3^{a-1}\Z_{3^a},
   \qquad
   U_{\Gamma}\!\left(\frac{2\pi}{3\sqrt3}\right)
   =P_{\beta_{a-1}3^{a-1}}.
\]
Consequently, every coset of $3^{a-1}\Z_{3^a}$ supports 3-vertex MST in
$\Gamma$.  If
$\beta_{a-1}=0$, or if $p>3$, then $\mathcal T=\{0\}$.
\end{corollary}

\begin{proof}
\noindent\textbf{Case 1.} $p=3$.
Write $3^a=3m$, where
$m=3^{a-1}$.  In Theorem~\ref{thm:all-orders-classification}, the condition
$X_0=\{m\}$ is equivalent to $\beta_{a-1}\ne0$.  The conductor-$3$ part of
that theorem gives the stated value of $\mathcal T$ and the transition matrix
when $\beta_{a-1}\ne0$, and gives $\mathcal T=\{0\}$ when
$\beta_{a-1}=0$.

\smallskip
\noindent\textbf{Case 2.} $p>3$.
Theorem~\ref{thm:global-conductor} excludes PST between distinct vertices.
\end{proof}

\begin{remark}
For $p=3$, the all-positive choice
$\varepsilon=(1,\ldots,1)$ in Theorem~\ref{thm:ramified-family} is the
second family in \cite[Example~6]{GodsilZhang2024}.  Allowing all
$\varepsilon\in\{\pm1\}^a$ gives $2^a$ connected labeled oriented
circulant graphs on $\Z_{3^a}$.  Corollary~\ref{cor:ramified-family-mst}
shows that every one of them has transfer subgroup
$3^{a-1}\Z_{3^a}$; the sign $\varepsilon_{a-1}$ determines the direction of
the translation attained at $2\pi/(3\sqrt3)$.  The two $2$-adic families in
\cite[Example~6]{GodsilZhang2024} are covered by
Theorem~\ref{thm:power-two-four} and by the conductor-$4$ case of
Theorem~\ref{thm:all-orders-classification}.
\end{remark}

\subsection{Examples}

The following examples illustrate how the conditions on $X_0$ and $X_1$ in
Theorem~\ref{thm:all-orders-classification} determine $\mathcal T$.  The
first example has both a ramified factor and a factor coprime to the
conductor.  The second compares two oriented circulant graphs of the same
order with $|\mathcal T|=2$ and $|\mathcal T|=4$, respectively.

\begin{example}\label{ex:eighteen}
Let $\Delta=3$, $n=18$, and $m=6$.  Let $X=\{1,6\}$, and let
$\sigma:X\to\{\pm1\}$ be given by $\sigma(1)=\sigma(6)=1$.  Set
\[
    C(X,\sigma)=\{1,6,7,13\}.
\]
Let $\Gamma=\mathcal I_{18,\chi_3}(X,\sigma)$ be the oriented circulant graph
with connection set $C(X,\sigma)$.  Since $1\in C(X,\sigma)$, $\Gamma$ is
connected.  Here $m=3\cdot2$ and
$X_0=\{6\}=\{m\}$.  Theorem~\ref{thm:all-orders-classification} gives
\[
   \mathcal T=6\Z_{18},
   \qquad
    U_{\Gamma}\!\left(\frac{2\pi}{3\sqrt3}\right)=P_6.
\]
Consequently, each of the six triples
\[
   \{b,b+6,b+12\},\qquad 0\le b<6,
\]
supports 3-vertex MST in $\Gamma$.
\end{example}

\begin{example}\label{ex:sixteen-comparison}
Let $\Delta=4$, $n=16$, and $m=4$.  Let
\[
    X^{(1)}=\{1,4\},
\]
and let $\sigma^{(1)}:X^{(1)}\to\{\pm1\}$ be the constant function with
value $1$.  Set
\[
    C^{(1)}=C(X^{(1)},\sigma^{(1)})=\{1,4,5,9,13\}.
\]
Let $\Gamma_1=\mathcal I_{16,\chi_4}(X^{(1)},\sigma^{(1)})$ be the oriented
circulant graph with connection set $C^{(1)}$.  Applying the notation
preceding Theorem~\ref{thm:all-orders-classification} to $X^{(1)}$ gives
$X^{(1)}_0=\{4\}$ and $X^{(1)}_1=\varnothing$, so
\[
    \mathcal T=\{0,8\},
    \qquad U_{\Gamma_1}(\pi/2)=P_8.
\]
Thus $\Gamma_1$ has antipodal PST but no 3-vertex MST.  For the second graph,
let
\[
    X^{(2)}=\{1,2,4\},
\]
let $\sigma^{(2)}:X^{(2)}\to\{\pm1\}$ be the constant function with value
$1$, and set
\[
    C^{(2)}=C(X^{(2)},\sigma^{(2)})=\{1,2,4,5,9,10,13\}.
\]
Let $\Gamma_2=\mathcal I_{16,\chi_4}(X^{(2)},\sigma^{(2)})$ be the oriented
circulant graph with connection set $C^{(2)}$.  Applying the same notation to
$X^{(2)}$ gives $X^{(2)}_1=\{2\}=\{m/2\}$, and Theorem~\ref{thm:all-orders-classification} gives
\[
    \mathcal T=4\Z_{16},
    \qquad U_{\Gamma_2}(\pi/4)=P_4.
\]
Hence each of the four sets
\[
   \{b,b+4,b+8,b+12\},\qquad 0\le b<4,
\]
supports 4-vertex MST in $\Gamma_2$.  Both $\Gamma_1$ and $\Gamma_2$ are
connected because their connection sets contain $1$.
\end{example}

\subsection{Consequences}

The classification has several immediate consequences for connected graphs,
transfer times, and the structure of the maximum MST sets.

\begin{corollary}\label{cor:connected-orders}
\begin{enumerate}[label=\textup{(\roman*)}]
\item There exists a connected oriented circulant graph
$\Gamma=\mathcal I_{3m,\chi_3}(X,\sigma)$ with 3-vertex MST if and only if
$m=1$ or $3\mid m$.
\item There exists a connected oriented circulant graph
$\Gamma=\mathcal I_{8m,\chi_8}(X,\sigma)$ with PST between distinct vertices
if and only if $m=1$ or $2\mid m$.
\item There exists a connected oriented circulant graph
$\Gamma=\mathcal I_{4m,\chi_4}(X,\sigma)$ with PST between distinct vertices
if and only if $m=1$ or $2\mid m$.  There exists such a connected graph with
4-vertex MST if and only if $m=2$ or $4\mid m$.
\end{enumerate}
Equivalently, the possible orders are $n=3$ or $9\mid n$ for
3-vertex MST, $n=8$ or $16\mid n$ for conductor-$8$ PST, and $n=4$ or
$8\mid n$ for conductor-$4$ PST.  Connected 4-vertex MST occurs at
conductor $4$ exactly when $n=8$ or $16\mid n$.
\end{corollary}

\begin{proof}
Write $m=p^a u$ as in Theorem~\ref{thm:all-orders-classification}.  If
$a=0$, PST requires $X=\{m\}$.
Proposition~\ref{prop:connectedness} then shows that
$\mathcal I_{\Delta m,\chi}(X,\sigma)$ is connected
exactly when $m=1$.  If $a\ge1$, the choice
\[
   X=\{m,1\}
\]
satisfies $X_0=\{m\}$ and gives a connected oriented circulant graph
$\mathcal I_{\Delta m,\chi}(X,\sigma)$ for arbitrary choices of the two
signs.  This establishes the assertions for conductors $3$ and $8$ and the
PST assertion for conductor $4$.

For 4-vertex MST, Theorem~\ref{thm:all-orders-classification} also requires
$a\ge1$ and $X_1=\{m/2\}$.  If $a=1$, these conditions force
$\gcd\{d:d\in X\}=u$, so connectedness is possible exactly when $u=1$, or
equivalently $m=2$.  If $a\ge2$, the choice
\[
   X=\{m,m/2,1\}
\]
is connected and satisfies both required conditions.  This is equivalent to
$4\mid m$.
\end{proof}

The next corollary lists all positive transfer times in the cases of
Theorem~\ref{thm:all-orders-classification} for which
$\mathcal T\ne\{0\}$.

\begin{corollary}\label{cor:all-order-times}
Let $\Gamma=\mathcal I_{\Delta m,\chi}(X,\sigma)$ be a nonempty oriented
circulant graph satisfying one of the cases of
Theorem~\ref{thm:all-orders-classification} for which
$\mathcal T\ne\{0\}$.

If $\Delta=3$ and $X_0=\{m\}$, put
\[
   \rho=
   \begin{cases}
      1,&\sigma(m)=1,\\
      2,&\sigma(m)=-1.
   \end{cases}
\]
For $\ell\in\{1,2\}$, we have $U(t)=P_{\ell m}$ if and only if
\[
   t=\frac{2\pi r}{3\sqrt3},
   \qquad r\in\Z_{>0},
   \qquad r\equiv\rho\ell\pmod3.
\]

If $\Delta=4$ and $\mathcal T=\{0,2m\}$, then $U(t)=P_{2m}$ if and only if
\[
   t=\frac{(2q+1)\pi}{2},
   \qquad q=0,1,2,\ldots.
\]
If $\Delta=4$ and $\mathcal T=m\Z_{4m}$, put
\[
   \rho=
   \begin{cases}
      1,&\sigma(m)=1,\\
      3,&\sigma(m)=-1.
   \end{cases}
\]
For $\ell\in\{1,2,3\}$, we have $U(t)=P_{\ell m}$ if and only if
\[
   t=\frac{\pi r}{4},
   \qquad r\in\Z_{>0},
   \qquad r\equiv\rho\ell\pmod4.
\]

If $\Delta=8$ and $X_0=\{m\}$, then $U(t)=P_{4m}$ if and only if
\[
   t=\frac{(2q+1)\pi}{\sqrt8},
   \qquad q=0,1,2,\ldots.
\]
There is no PST between distinct vertices in the remaining cases.
\end{corollary}

\begin{proof}
In every case with $\mathcal T\ne\{0\}$,
Theorem~\ref{thm:all-orders-classification} gives $g=1$.
Theorem~\ref{thm:transfer-structure} then gives the minimum periods
$2\pi/\sqrt3$, $\pi$, and $2\pi/\sqrt8$ for conductors $3$, $4$, and $8$,
respectively.  The same theorem shows that the translation
times are equally spaced within each period.  The transition matrices in
Theorem~\ref{thm:all-orders-classification} and their powers determine the
displayed residue classes.
\end{proof}

\begin{corollary}\label{cor:all-order-mst-sets}
Let $\Gamma=\mathcal I_{\Delta m,\chi}(X,\sigma)$ be a nonempty oriented
circulant graph satisfying one of the cases of
Theorem~\ref{thm:all-orders-classification} for which
$\mathcal T\ne\{0\}$.  The maximum sets supporting pairwise PST in $\Gamma$
are as follows.
\begin{enumerate}[label=\textup{(\roman*)}]
\item For $\Delta=3$, they are the $m$ triples
\[
   \{b,b+m,b+2m\},\qquad 0\le b<m.
\]
\item For $\Delta=4$ and $\mathcal T=m\Z_{4m}$, they are the $m$ sets
\[
   \{b,b+m,b+2m,b+3m\},\qquad 0\le b<m.
\]
If $\mathcal T=\{0,2m\}$, the maximum sets are the $2m$ antipodal pairs
$\{b,b+2m\}$, where $0\le b<2m$.
\item For $\Delta=8$, they are the $4m$ antipodal pairs
$\{b,b+4m\}$, where $0\le b<4m$.
\end{enumerate}
\end{corollary}

\begin{proof}
By Theorem~\ref{thm:transfer-structure}, the maximum sets are precisely the cosets
of $\mathcal T$.  Substitution of the three possible values of $\mathcal T$
in Theorem~\ref{thm:all-orders-classification} gives the listed sets.
\end{proof}

\begin{corollary}\label{cor:radicand-transfer}
Let $\delta$ be a positive square-free integer, let
$\Delta=\Delta(\delta)$ be the conductor defined in Section~2, and let $\chi$
be the corresponding odd primitive quadratic character.  There exists a
multiple $n$ of $\Delta$ and an oriented circulant graph
$\Gamma=\mathcal I_{n,\chi}(X,\sigma)$ with PST
between distinct vertices if and only if
\[
   \delta\in\{1,2,3\}.
\]
For such a graph, let $\mathcal T$ be its transfer subgroup.  For
$\delta=1$, the cardinality of
$\mathcal T$ can be $2$ or $4$; for $\delta=2$, it can only be $2$; and for
$\delta=3$, it can only be $3$.  Each possibility is attained.
The same existence statement holds with PST replaced by PGST.
\end{corollary}

\begin{proof}
By the definition of $\Delta(\delta)$, the conductors $4$, $8$, and $3$
correspond to the square-free radicands $1$, $2$, and $3$, respectively.
Theorem~\ref{thm:global-conductor} excludes every other conductor.
Theorem~\ref{thm:all-orders-classification} gives examples attaining each stated
cardinality.  The final assertion follows from Theorem~\ref{thm:transfer-structure}.
\end{proof}

\begin{remark}
For $\Delta=4$, we have $\sqrt{\Delta}=2$, and the graphs
$\mathcal I_{4m,\chi_4}(X,\sigma)$ are precisely the integral oriented
circulant graphs.  The conductor-$4$ part of Theorem~\ref{thm:all-orders-classification} restates the classification in
\cite{Song2024IOC} in terms of the transfer subgroup $\mathcal T$.
The integral mixed case was settled in \cite{SongLin2024Mixed}.
\end{remark}

Theorem~\ref{thm:all-orders-classification} and
Table~\ref{tab:all-orders-classification} give the classification, while
Corollaries~\ref{cor:all-order-times} and \ref{cor:all-order-mst-sets}
describe its transfer-time and MST consequences.  The final subsection
counts the resulting graphs and introduces no additional state-transfer
cases.

\subsection{Enumeration and small orders}

Fix an odd primitive quadratic character $\chi$ of conductor $\Delta$ and an
order $n=\Delta m$.  The connection-set description in Definition~\ref{def:main-graph} reduces enumeration to independent choices on admissible
gcd-classes.  We count labeled oriented circulant graphs
$\mathcal I_{n,\chi}(X,\sigma)$, equivalently their connection sets, rather
than isomorphism classes.  In particular, when $X\ne\varnothing$, the
isomorphic graphs associated with $\sigma$ and $-\sigma$ in Remark~\ref{rem:sign-reversal} are counted separately.  Let $\mathrm d(m)$ denote
the number of positive divisors of $m$.

\begin{proposition}\label{prop:counting}
Let $n=\Delta m$.  Definition~\ref{def:main-graph} gives
\[
    3^{\mathrm d(m)}
\]
distinct connection sets $C(X,\sigma)$, including the empty connection set.
Consequently, there are $3^{\mathrm d(m)}$ labeled oriented circulant graphs
$\Gamma=\mathcal I_{n,\chi}(X,\sigma)$.  The number of connected labeled
oriented circulant graphs is
\begin{equation}\label{eq:connected-count}
   \sum_{e\mid m}\mu(e)
      \left(3^{\mathrm d(m/e)}-1\right).
\end{equation}
\end{proposition}

\begin{proof}
The admissible gcd-class indices are precisely the divisors of $m$.  For each
such divisor, the corresponding class may be omitted or included with either
orientation, giving three choices.  For $e\mid m$, there are
$3^{\mathrm d(m/e)}-1$ nonempty choices whose indices are all divisible by
$e$.  Proposition~\ref{prop:connectedness} identifies the graphs with
connected components indexed by these common divisors, and M\"obius inversion
gives \eqref{eq:connected-count}.
\end{proof}

\begin{corollary}\label{cor:all-order-count}
Fix $\Delta\in\{3,4,8\}$, let $\chi$ be the odd primitive quadratic character
of conductor $\Delta$, and write $m=p^a u$ as in Theorem~\ref{thm:all-orders-classification}.  Among the labeled oriented circulant
graphs $\Gamma=\mathcal I_{\Delta m,\chi}(X,\sigma)$ arising from all
connection-set choices, including the empty connection set, exactly
\[
   2\cdot3^{a\mathrm d(u)}
\]
have PST between distinct vertices.
For $\Delta=3$, every graph counted above has 3-vertex MST; for
$\Delta=8$, every graph counted above has antipodal PST.  For
$\Delta=4$ and $a\ge1$, exactly
\[
   4\cdot3^{(a-1)\mathrm d(u)}
\]
have 4-vertex MST, and the remaining
$2\cdot3^{a\mathrm d(u)}-4\cdot3^{(a-1)\mathrm d(u)}$ satisfy
$|\mathcal T|=2$.
When $a=0$, the two labeled conductor-$4$ oriented circulant graphs with PST
also satisfy $|\mathcal T|=2$, and both are connected if and only if $m=1$.
Since $\mathrm d(m)=(a+1)\mathrm d(u)$, the proportion among all
connection-set choices, including the empty graph, is
\[
   \frac{2}{3^{\mathrm d(u)}}.
\]
This proportion is independent of $a$.  Among the nonempty labeled oriented
circulant graphs, the corresponding proportion is
\[
   \frac{2\cdot3^{a\mathrm d(u)}}
        {3^{(a+1)\mathrm d(u)}-1}.
\]
\end{corollary}

\begin{proof}
The condition $X_0=\{m\}$ leaves two choices for $\sigma(m)$ and excludes the
other $\mathrm d(u)-1$ possible indices in $X_0$.  The remaining
$a\mathrm d(u)$ indices are arbitrary, giving $2\cdot3^{a\mathrm d(u)}$
labeled oriented circulant graphs.  For 4-vertex MST at conductor $4$, the
additional condition $X_1=\{m/2\}$ leaves two choices for
$\sigma(m/2)$, excludes the other $\mathrm d(u)-1$ possible indices in
$X_1$, and leaves $(a-1)\mathrm d(u)$ indices arbitrary.  Multiplication of
these independent choices gives the stated formulas.
\end{proof}

\begin{corollary}\label{cor:connected-all-order-count}
Fix $\Delta\in\{3,4,8\}$, let $\chi$ be the odd primitive quadratic character
of conductor $\Delta$, and let $m=p^a u$ be as in Theorem~\ref{thm:all-orders-classification}.  Let $N_{\mathrm{conn}}$ be the number
of connected labeled oriented circulant graphs
$\Gamma=\mathcal I_{\Delta m,\chi}(X,\sigma)$ having PST between distinct
vertices.  Then
\[
N_{\mathrm{conn}}=
\begin{cases}
2,&a=0\text{ and }u=1,\\
0,&a=0\text{ and }u>1,\\
2\displaystyle\sum_{v\mid u}\mu(v)
\left(3^{a\mathrm d(u/v)}-3^{(a-1)\mathrm d(u/v)}\right),&a\ge1.
\end{cases}
\]
For conductor $4$, let $N_4$ be the number of connected labeled oriented
circulant graphs $\Gamma=\mathcal I_{4m,\chi_4}(X,\sigma)$ with 4-vertex MST.
Then
\[
N_4=
\begin{cases}
0,&a=0,\\
4,&a=1\text{ and }u=1,\\
0,&a=1\text{ and }u>1,\\
4\displaystyle\sum_{v\mid u}\mu(v)
\left(3^{(a-1)\mathrm d(u/v)}-3^{(a-2)\mathrm d(u/v)}\right),&a\ge2.
\end{cases}
\]
Thus $N_{\mathrm{conn}}-N_4$ connected labeled conductor-$4$ oriented
circulant graphs satisfy $|\mathcal T|=2$.
\end{corollary}

\begin{proof}
The condition $X_0=\{m\}$ forces $m$ to be in $X$ and leaves the
$a\mathrm d(u)$ divisors
\[
   p^b e,\qquad 0\le b<a,\qquad e\mid u,
\]
free.  For $f\mid m$, let $N(f)$ be the number of these free divisors that
are divisible by $f$.  There are $2\cdot3^{N(f)}$ labeled oriented circulant
graphs for which every element of $X$ is divisible by $f$.  Proposition~\ref{prop:connectedness} and M\"obius inversion therefore give
\[
   N_{\mathrm{conn}}=2\sum_{f\mid m}\mu(f)3^{N(f)}.
\]
If $a\ge1$, M\"obius inversion leaves only the terms $f=v$ and $f=pv$, where
$v\mid u$, and
\[
   N(v)=a\mathrm d(u/v),
   \qquad
   N(pv)=(a-1)\mathrm d(u/v).
\]
Substitution gives the formula for $N_{\mathrm{conn}}$; the case $a=0$
follows from $\sum_{v\mid u}\mu(v)$.

For 4-vertex MST, the additional condition $X_1=\{m/p\}$ forces both $m$
and $m/p$, with four choices for their signs.  If $a\ge2$, the free divisors
are $p^b e$ with $0\le b\le a-2$ and $e\mid u$.  For $v\mid u$, the numbers
of free divisors divisible by $v$ and by $pv$ are, respectively,
$(a-1)\mathrm d(u/v)$ and $(a-2)\mathrm d(u/v)$.  The same M\"obius
inversion gives the displayed formula for $N_4$.  When $a=1$, there are no
free divisors, and the forced set is connected exactly when $u=1$.
\end{proof}

Let $\Gamma=\mathcal I_{n,\chi}(X,\sigma)$ be a nonempty oriented circulant
graph, and let $h$ be the integer in \eqref{eq:transfer-gcd}.  By
Theorem~\ref{thm:transfer-structure}, $h$ is the maximum cardinality of a set
supporting pairwise PST in $\Gamma$.  Table~\ref{tab:small-orders} records the
connected labeled oriented circulant graphs at selected small orders.  Its
first two columns give the square-free radicand $\delta$ and the corresponding
conductor $\Delta=\Delta(\delta)$.  An entry $h:c$ in the last column means
that exactly $c$ such graphs have the indicated value of $h$; in particular,
$h=1$ means that PST does not occur between distinct vertices.  The first
block covers all conductor and transfer-subgroup cases in
Theorem~\ref{thm:all-orders-classification}: $h=3$ at conductor $3$, $h=2$ or
$4$ at conductor $4$, and $h=2$ at conductor $8$.  It also includes orders
for which every connected graph has $h=1$.  The second block gives examples
from both forms of conductors of odd primitive quadratic characters in
Proposition~\ref{prop:delta} that are excluded by
Theorem~\ref{thm:global-conductor}.

\begin{table}[htbp]
\centering
\small
\caption{Representative connected labeled oriented circulant graphs,
classified by the maximum cardinality $h$ of a set supporting pairwise PST}
\label{tab:small-orders}
\vspace{0.5em}
\begin{tabular}{cccccc}
\toprule
$\delta$ & $\Delta$ & $n$ & $\mathrm d(n/\Delta)$ &
number of graphs & distribution by $h$ \\
\midrule
\multicolumn{6}{l}{\emph{Conductors $3$, $4$, and $8$}}\\
$3$ & $3$ & $3$  & $1$ & $2$  & $3:2$ \\
$3$ & $3$ & $6$  & $2$ & $6$  & $1:6$ \\
$3$ & $3$ & $9$  & $2$ & $6$  & $1:2$, $3:4$ \\
$3$ & $3$ & $18$ & $4$ & $66$ & $1:54$, $3:12$ \\
$3$ & $3$ & $27$ & $3$ & $18$ & $1:6$, $3:12$ \\
$1$ & $4$ & $4$  & $1$ & $2$  & $2:2$ \\
$1$ & $4$ & $8$  & $2$ & $6$  & $1:2$, $4:4$ \\
$1$ & $4$ & $12$ & $2$ & $6$  & $1:6$ \\
$1$ & $4$ & $16$ & $3$ & $18$ & $1:6$, $2:4$, $4:8$ \\
$1$ & $4$ & $24$ & $4$ & $66$ & $1:54$, $2:12$ \\
$1$ & $4$ & $48$ & $6$ & $630$ & $1:498$, $2:108$, $4:24$ \\
$2$ & $8$ & $8$  & $1$ & $2$  & $2:2$ \\
$2$ & $8$ & $16$ & $2$ & $6$  & $1:2$, $2:4$ \\
$2$ & $8$ & $24$ & $2$ & $6$  & $1:6$ \\
$2$ & $8$ & $32$ & $3$ & $18$ & $1:6$, $2:12$ \\
$2$ & $8$ & $48$ & $4$ & $66$ & $1:54$, $2:12$ \\
\addlinespace[3pt]
\multicolumn{6}{l}{\emph{Other conductors of odd quadratic characters}}\\
$7$  & $7$  & $7$  & $1$ & $2$ & $1:2$ \\
$7$  & $7$  & $49$ & $2$ & $6$ & $1:6$ \\
$11$ & $11$ & $11$ & $1$ & $2$ & $1:2$ \\
$15$ & $15$ & $15$ & $1$ & $2$ & $1:2$ \\
$5$  & $20$ & $20$ & $1$ & $2$ & $1:2$ \\
$5$  & $20$ & $40$ & $2$ & $6$ & $1:6$ \\
$6$  & $24$ & $24$ & $1$ & $2$ & $1:2$ \\
\bottomrule
\end{tabular}
\end{table}

The rows $(\Delta,n)=(3,6),(4,12),(8,24)$ show that a conductor admitting PST
need not yield PST in a connected graph at every multiple of the conductor.
At $(\Delta,n)=(4,48)$, all three possibilities $h=1,2,4$ occur.  The rows
with $\Delta=20$ represent the square-free radicand $\delta=5$, since
$\Delta(5)=20$.  All entries follow from
\eqref{eq:general-spectrum} and \eqref{eq:transfer-gcd} by exact integer
arithmetic.

The verification script \texttt{verify\_enumeration.py} enumerates the three
choices for each admissible gcd-class, computes the integer coefficient sequence in
\eqref{eq:general-spectrum}, divides it by its gcd, evaluates $h$ from
\eqref{eq:transfer-gcd}, and tests connectedness using
Proposition~\ref{prop:connectedness}.  No floating-point computation is
used.
The rows with $\Delta=3,4,8$ follow by specializing
Corollaries~\ref{cor:all-order-count} and
\ref{cor:connected-all-order-count}.
For every conductor in the second block,
Theorem~\ref{thm:global-conductor} excludes PST between distinct vertices.
The conductors $7,11,15$ represent the first alternative in
Proposition~\ref{prop:delta}, and $20,24$ represent the second.  Direct evaluation of
\eqref{eq:general-spectrum} and \eqref{eq:transfer-gcd} using exact integer
arithmetic gives the displayed counts.

\section{Concluding remarks}

The periodicity theorem of Godsil and Lato \cite[Theorem~6.1]{GodsilLato2020}
and the connection-set characterization of Godsil and Zhang
\cite[Theorem~8.1]{GodsilZhang2024} reduce PST on nonempty oriented circulant
graphs to the quadratic-character family
$\mathcal I_{n,\chi}(X,\sigma)$; see
Corollary~\ref{cor:pst-character-family}.  The character-theoretic condition
of Chan et al.\ \cite[Theorem~3.3]{ChanPantangiRazafimahatratraSin2026}
provides a general test for oriented normal Cayley graphs.  The two additional
ingredients used here are the explicit, Fourier-indexed eigenvalue formula in
Theorem~\ref{thm:general-spectrum}, valid also at ramified orders, and
Theorem~\ref{thm:transfer-structure}, which determines the vertices joined by
PST for square-root spectra.

These ingredients restrict PST between distinct vertices to conductors
$3$, $4$, and $8$ and give the exact conditions on $X$ at every multiple of
these conductors.  The square-free radicands are therefore precisely $1$,
$2$, and $3$.  The conductor-$4$ case agrees with the known integral
oriented classification \cite{Song2024IOC}; the nonintegral cases occur at
conductors $3$ and $8$.  The general four-vertex bound follows from
\cite[Theorem~4.1]{ChanPantangiRazafimahatratraSin2026}; the cyclic
classification obtained here determines exactly when $\mathcal T$ has
cardinality $2$, $3$, or $4$, together with all PST times, maximum MST sets,
connected orders, and enumeration formulas.

For general finite abelian groups, the connection-set structure is known
\cite{GodsilZhang2024}, and general character-theoretic conditions for PST
are available \cite{ChanPantangiRazafimahatratraSin2026}.  The remaining
problem is to convert these results into an explicit classification of
the vertex sets on which pairwise PST occurs, the PST times, and the MST sets,
analogous to the cyclic classification obtained here.

\section*{Data availability statement}

No external datasets were used in this work.  The Python script
\texttt{verify\_enumeration.py}, supplied as supplementary material,
reproduces the enumeration results and Table~\ref{tab:small-orders}.

\end{document}